\documentclass[MNA]{mna} 

\SHORTTITLE{Neural Field Theory Topogprahic Development}
\TITLE{Analysis of Activity Dependent Development of Topographic Maps in Neural Field Theory with Short Time Scale Dependent Plasticity.}
\AUTHORS{
  Nicholas~Gale\footnote{University of Cambridge.
    \EMAIL{nmg41@cam.ac.uk}}
\and
 Jennifer~Rodger\footnote{University of Western Australia.
 	\EMAIL{jennifer.rodger@uwa.ed.au}}
\and
 Michael~Small\footnote{University of Western Australia
 	\EMAIL{michael.small@uwa.edu.au}}
\and
  Stephen~Eglen\footnote{University of Cambridge.
	\EMAIL{sje30@cam.ac.uk}}
}
\SHORTAUTHOR{N. Gale, J. Rodger, M. Small, S. Eglen}
\AMSSUBJ{92B20; 45K05; 42B37}
\SUBMITTED{August 25, 2021} 
\ACCEPTED{January 29, 2022} 
\ARXIVID{2107.13272}

\VOLUME{2}
\YEAR{2022}
\PAPERNUM{1}
\DOI{10.46298/mna.8390}

\ABSTRACT{Topographic maps are a brain structure connecting pre-synaptic and 
post-synaptic brain regions.  Hebbian-based plasticity mechanisms working in conjunction 
with spontaneous patterns of neural activity generated in the pre-synaptic regions play a 
critical role in appropriate topographic development. Studies performed in mouse have 
shown that these spontaneous patterns can exhibit complex spatial-temporal structures 
which existing models cannot incorporate. Neural field theories are appropriate modelling 
paradigms for topographic systems due to the dense nature of the connections between 
regions and can be augmented with a plasticity rule general enough to capture complex 
time-varying structures. 
	
We propose a theoretical framework for studying the development of topography in the 
context of complex spatial-temporal activity feed-forward from the pre-synaptic to 
post-synaptic regions. Analysis of the model leads to an analytic solution corroborating the 
conclusion that activity can drive the refinement of topographic projections. The analysis also 
suggests that biological noise is used in the development of topography to stabilise the 
dynamics. MCMC simulations are used to analyse and understand the differences in 
topographic refinement between wild-type and the $\beta2$ knock-out mutant in mice. The 
time scale of the synaptic plasticity window is estimated as $0.56$ seconds in this context 
with a model fit of $R^2 = 0.81$. }

\usepackage{graphicx}
\usepackage{caption}
\usepackage{subcaption}
\usepackage{float}
\usepackage{hyperref}

\begin{document}

\section{Introduction}
A topographic map is a ubiquitous brain structure which connects two brain regions: a 
pre-synaptic region and post-synaptic region \cite{Udin1988-by}. The structure is defined by 
the relationship that cells that are physically neighbouring in the pre-synaptic region will 
connect to physically neighbouring cells in the post-synaptic region and is the simplest 
instance of a topological map \cite{Bednar2016-lg}. Topographic maps have been shown to 
have remarkable regeneration and re-organisational properties but this study will focus only 
on their development \cite{Buonomano1998-wi, Merzenich1983-uz, Robertson1989-sm, 
Kaas1990-kw}. Historically, topographic development was thought to be mediated by either 
Hebbian activity-based mechanisms, or chemotactic signalling mechanisms but these now 
are thought to more typically work in tandem \cite{Cang2013-dw}.

The dense feed-forward connectivity pattern present in topographic systems make neural 
field theories (NFT) an attractive paradigm for modelling electrical activity patterns in 
topographic systems.  An NFT is a continuum model where the spiking activity of many neural 
inputs are averaged into a smoothly varying function over temporal and spatial locations. A 
theory of topographic development was proposed for an NFT which relied on static inputs in 
the pre-synaptic region resulting in activity patterns in the post-synaptic region stabilising to 
be time-independent and thus allowing a simple Hebbian plasticity rule to be applied 
\cite{Amari1977-gc}. The requirement for activity to stabilise before updating weight is also a 
feature of more general cortical map plasticity models \cite{Stevens2013-ly, Bednar2004-xl}. 
A more recent study used an NFT to model topography in the somatosensory cortex under 
thalamocortical plasticity using Oja's rule \cite{Detorakis2012-eh}. The assumption of static 
inputs limits the range of biological systems for which the theory developed in these works 
can apply in development. 

A proposed candidate model organism is the mouse retinotopic map: the set of connections 
that map retinal cells in the eye to cells in the superior colliculus (SC) \cite{Ito2018-ef, 
Seabrook2017-kx}. This is distinct from other topographic projections such as 
somatosensory and tonotopic maps \cite{Kaas1997-oz}. The mouse develops topography 
using three mechanisms: chemotaxis, competition, and activity based refinement 
\cite{McLaughlin2003-yy, Cang2013-dw}. The activity component of development involves 
three stages of spontaneously generated retinal waves which are thought to refine a coarse 
topography of dendritic arbours (grown from afferent neurons which are guided 
topographically by a combination of chemotaxis and competitive interactions) down into a 
precise point-to-point mapping \cite{Cang2013-dw, Bansal2000-ts, Maccione2014-ha}. 
Disruptions to the patterning of these waves have been explored by knocking out the 
nicotinic-acetylcholine receptor $\beta2$ which generates fast-spreading waves and thus a 
hyper-correlation -- where neurons are correlated at a far greater inter-neuron distance than 
in wild type -- between any two given retinal cells \cite{Stafford2009}. The effect of the 
$\beta2$ knock-out is to reduce the precision of the resulting topographic map: the 
receptive field of a given SC location is large with respect to wild-type 
\cite{Mrsic-Flogel2005-xp, McLaughlin2003-yy, Chandrasekaran2005-ug}.

Modelling efforts in this field have focused recently on predicting map structure of various mutants and a unified model of chemotaxis, activity, competitive mechanisms was shown to give the best account of the data \cite{Hjorth2015-le, Triplett2011-jk, Tikidji-Hamburyan2016-sn, Tsigankov2006-uy}. The mutants examined were predominantly genetic perturbations of the chemical gradients and therefore activity was not considered as a major focus. These models by construction are unable to capture the various spatio-temporal statistics of the retinal waves condensing them all into a single correlation measure as a function of SC-distance; a corollary is that they have not been able to reproduce the effect of the $\beta2$ knock-out when the correlation function is adjusted to match the knock-out \cite{Lyngholm2019-fs}. While historical models have allowed for the incorporation of spatial patterning in the input stimulus, they do not consider time variations in stimulus at a time-scale below that of plasticity implicitly assuming all transient neuronal information (such as spatio-temporally patterned stimulus waves) is averaged out \cite{Willshaw1976-ew, Kohonen1982-nd, Stevens2013-ly}. Recent efforts in a separate unified model which can incorporate dynamic activity were not able to quantitatively account for the $\beta2$ mutant data and are too computationally demanding for rigorous statistical analysis \cite{Godfrey2009-ts}. There is a need for theory which can analyse and predict the effects of time-varying stimuli on the organisational structure of maps.

In this paper we aim to develop theory for modelling the development of topographic 
systems which can incorporate complex spatial-temporal patterns of activity, such as those 
seen in mouse. A candidate theoretical framework of Hebbian-based plasticity that can 
incorporate time-signatures of activity, spike timing dependent plasticity (STDP), has been 
developed for NFT \cite{Robinson2011-ve, Abbott2000-gl}.  The framework is continuous, 
rather than discrete, which allows us to investigate synaptic efficacy between locations 
rather than modelling individual synapses. We shall demonstrate that NFT can support the 
refinement and establishment of precise topography via waves of propagating activity and 
biologically reasonable Hebbian learning rules and therefore establish it as a useful model to 
study the development of topographic systems. Moreover, we will validate the model against 
the $\beta2$ knock-out and make predictions about the time-scale on which the Hebbian 
activity operates. A glossary of symbols that will be used throughout the paper is shown in 
Table \ref{table:glossary}.

\begin{table}
	\centering
	\begin{tabular}{| l || l |}
		\hline
		Symbol & Description \\
		\hline
		$S$ &Feed-forward kernel: pre-to-post regions \\
		$W$ &Recurrent kernel: post-to-post region \\
		$H$ &Synaptic evolution by spike-time envelope \\
		$Q$ &Map of membrane potential to spike-rate\\
		$u/U$ &Post-synaptic membrane potential/rates\\
		$a/A$ &Pre-synaptic membrane potential/rates\\
		$h$ &Form of activity waves\\
		$\Theta$ &Heaviside Theta function\\
		$\delta$ &Dirac-Delta distribution\\
		$\eta$ &Representation of white noise\\
		\hline
	\end{tabular}
	\caption{Symbols which are used to represent biological objects and/or processes. The parameters which are used to specify each functional form are omitted and detailed in later sections. \label{table:glossary}}
\end{table}

\section{Model}
We will choose a simple model architecture that closely imitates the systems of interest: input from a continuous pre-synaptic field of nerve cells stimulates activity in a continuous post-synaptic field of nerve cells via a collection of feed-forward connections. These feed-forward connections will evolve under a plasticity rule governed by the spatio-temporal relations between the input and induced activity in the pre-synaptic and post-synaptic fields respectively. The activity in the post-synaptic field will be supported by inhibitory and excitatory sets of isotropic recurrent (or lateral) connections which, for simplicity, we shall assume to be static; for a description of non-isotropic recurrent kernels refer to \cite{Graben2014-pm, Schwappach2015-jy}. Changes in the feed-forward connections are dictated by firing activity in the pre-synaptic and post-synaptic fields. The activity dynamics in the post-synaptic field will be modelled by a neural field equation which couples a membrane potential and firing activity spatio-temporally. The pre-synaptic field activity could be modelled the same way but because there is no feed-back from the post-synaptic field it is sufficient to simply instantiate it which can be motivated by experimental spiking data \cite{Meister1991-mu}. The model architecture is summarised in Figure \ref{model-summary} and we shall now explicitly lay out the details of the model. 
\begin{figure}[h!]
	\centering
	\includegraphics[width=0.7\textwidth]{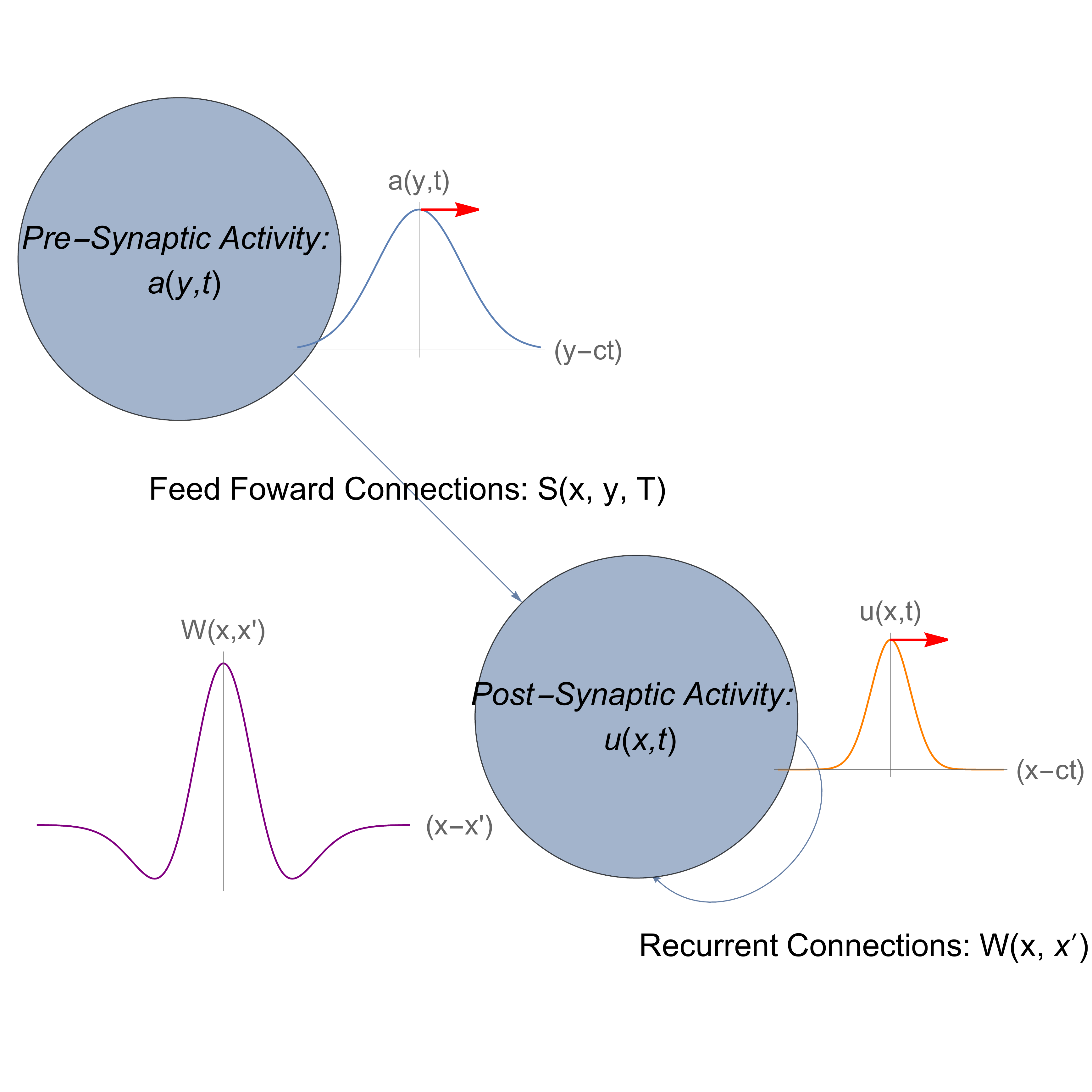}
	\caption{The connections and directionality of the model: activity is feed-forward from the 
	pre-synaptic region by the structure of interest and is spatial-temporally propagated by a 
	time-differential operator and spatial convolution of inhibitory and excitatory recurrent 
	connections. We show cartoons of typical propagating activity patterns and recurrent 
	connections but not feed-forward connections as determining these are the object of this 
	study. The generated signal in the post-synaptic region and the driving signal in the 
	pre-synaptic region are then convolved with a plasticity window to inform the synaptic 
	changes on a slow time scale which is indicated by the variable $T = \epsilon t$ for some 
	small $\epsilon$. \label{model-summary}}
\end{figure}
\paragraph{Representation of Topography}
We need to establish what we mean by topography in the continuous sense. We aim to 
preserve two things: the neighbourhood projection, and the excitatory feed-forward nature of 
the network. To preserve the neighbour-neighbour relation the connections,  here referred to 
as a synaptic distribution, labelled $S$, and measured in synapses per mm$^2$,  should take 
the form:
\begin{equation}
	S(x,y,T) = S(|x-p(y)- \rho|, T),
\end{equation}
where $p(y)$ is some monotonically increasing function and $\rho$ is some constant to indicate that a coordinate shift still permits a topographic mapping. The excitatory feed-forward nature means that a patch of activation in the pre-synaptic field should activate a local patch of the post-synaptic field associated with its topographically projected location. Therefore, $S$ should decay quickly at infinity, be positive at the topographically projected location, and have a finite (small) radius at which it transitions to being negative. Alternatively, it can be strictly positive but quickly decaying such that it never over-powers the recurrent inhibitory connections; see Figure \ref{fig:topographicexamples}.

\begin{figure}[h!]
	\begin{subfigure}{0.4\textwidth}
		\centering
		\includegraphics[width=\textwidth]{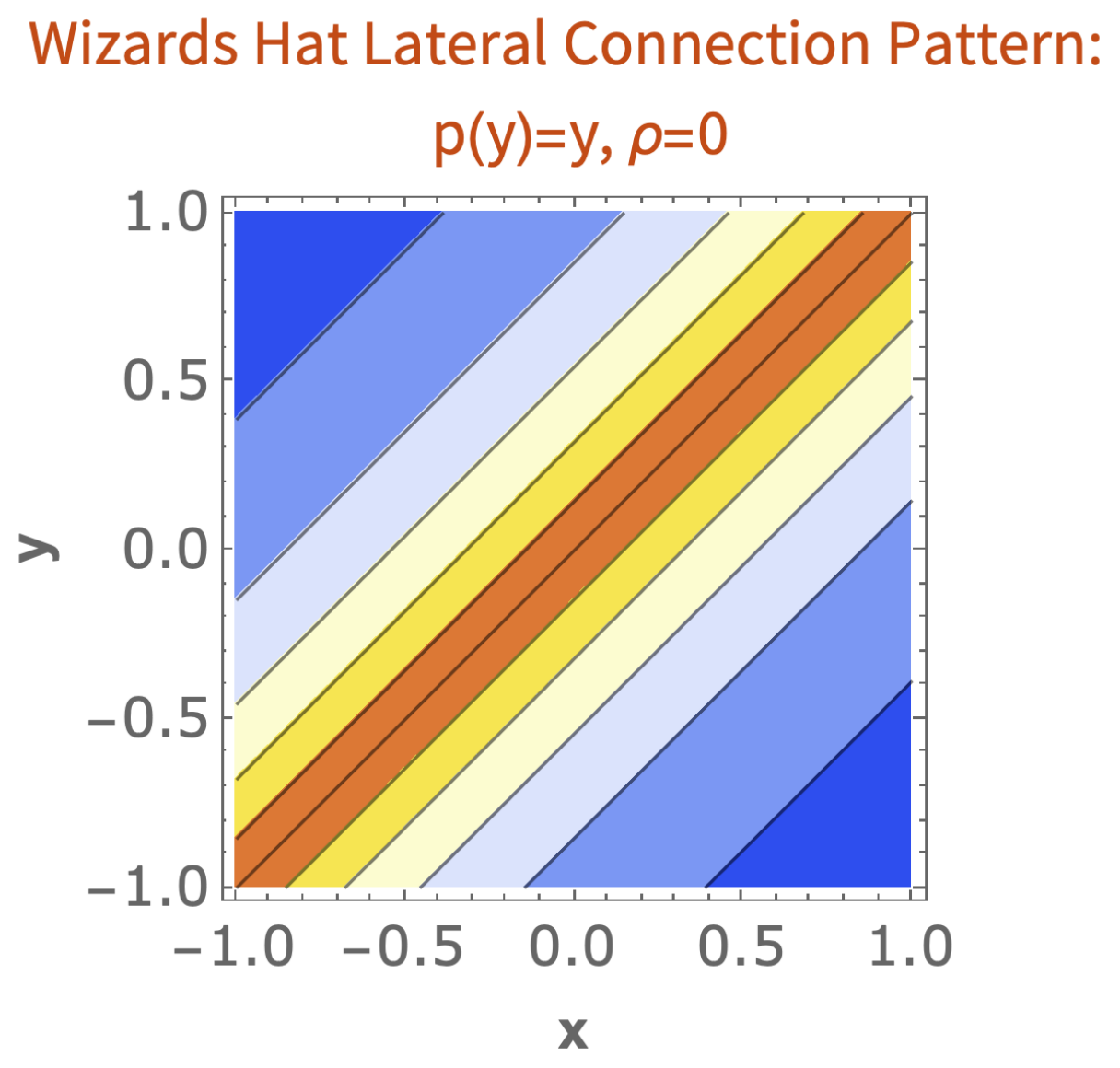}
		\caption{}
	\end{subfigure}
	\begin{subfigure}{0.52\textwidth}
		\centering
		\includegraphics[width=\textwidth]{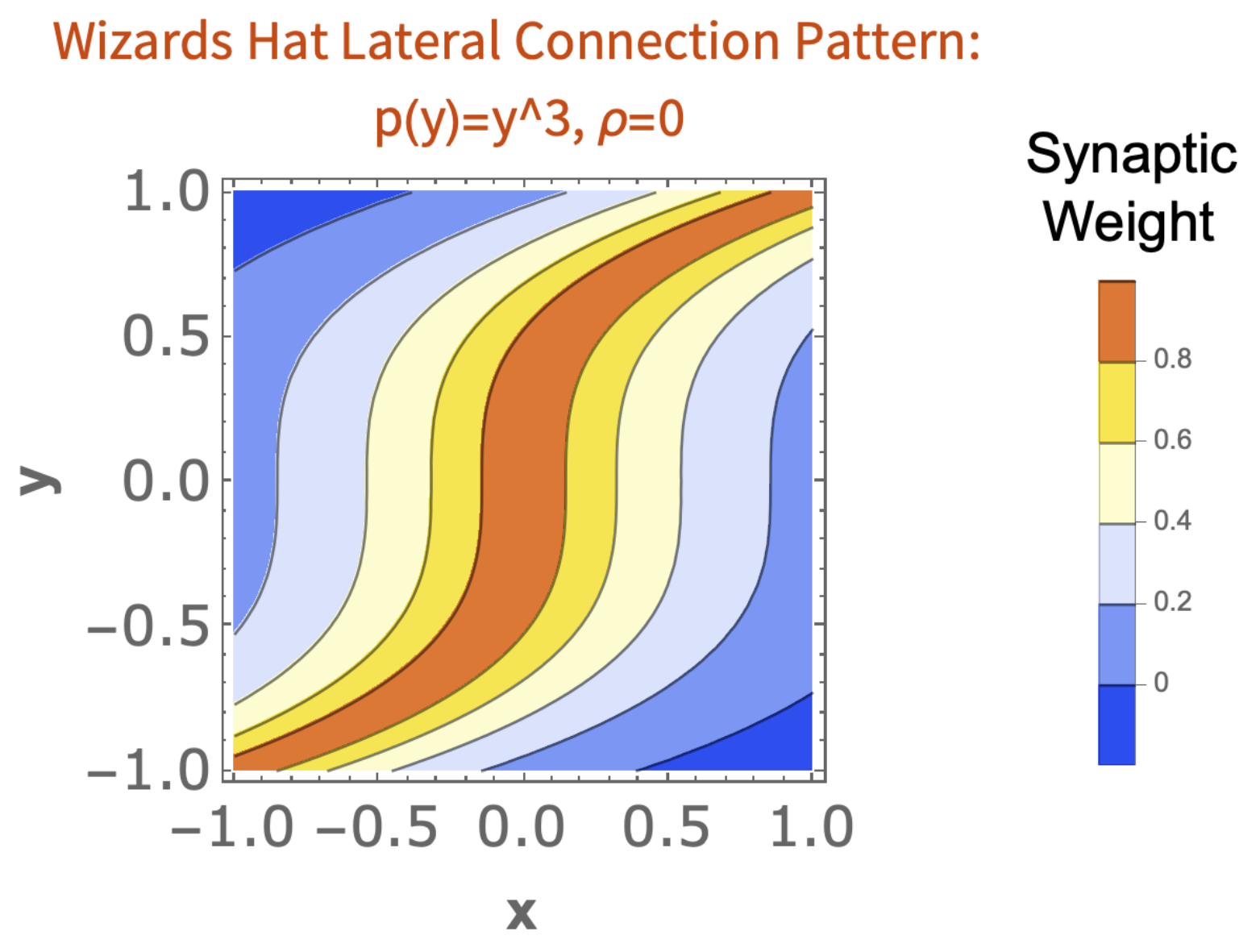}
		\caption{}
	\end{subfigure}
	\caption{\label{fig:topographicexamples}  Two examples of topographic organisation using a Wizard hat style function: a) shows a linear relationship between axes and for while b) shows a cubic relationship between axes. Both are topographic but a) has an even representation of the pre-synaptic field across the post-synaptic field while b) compresses the representation at the boundary and enlarges the interior.}
\end{figure}
\paragraph{Neural Field Theory}
We shall choose the NFT formulation proposed by Amari \cite{Amari1977-gc} 	 which 
considers both 
excitatory and inhibitory connections in the same kernel $W$. We shall 
consider a kernel $S$ that also couples the pre-synaptic and post-synaptic regions. It is this 
kernel in which we aim to demonstrate the evolution of topography. We shall denote the 
electrical activity, measured in millivolts (mV), of the pre-synaptic field by $a(x,t)$ and in the 
post-synaptic field by $u(x,t)$, and choose the firing rate function to be a sigmoid-logistic 
function:
\begin{equation}
	Q(u)=\frac{Q_\text{max}}{1+\exp(-\beta(u-\theta))},
\end{equation} 
where $\beta$ and $\theta$ dictate the steepness of the curve and the threshold respectively, and $Q_{\text{max}}$ is the maximal firing rate; { the units of each are $mV^{-1}$, $mV$, and spikes per second and can be found in Table \ref{table:parameters}}. The activity dynamics are then governed by the internal dynamics mediated by $W$ and the input provided through the pre-synaptic field, $Q(a)$, and its transfer through $S$:
\begin{align}
	\begin{split}
		u(x,t)&
		+\tau \frac{\partial u(x,t)}{\partial t} 
		- \int_{-\infty}^{\infty} W(x, x^\prime) Q(u(x^\prime,t)) dx^\prime = \int_{-\infty}^{\infty}S(x,y, T) Q(a(y,t)) dy.
	\end{split}
	\label{eq:amari}
\end{align}
Note that the time variable $T$ is on a much slower time scale which is realised by setting 
$t=\epsilon T$ for $0<\epsilon \ll 1$. For the purposes of solving  \eqref{eq:amari} these 
connections can be considered effectively constant. We assume for simplicity that the 
recurrent connections $W$ remain constant throughout the course of synaptic development 
and are homogenous. Following Robinson \cite{Robinson2011-ve}  a plasticity window is 
defined as a rapidly 
decaying envelope $H$ that weights the cross-correlation of the input and response signals 
in a population in the same fashion as biologically-inspired plasticity rules weight individual 
spikes of a neuron \cite{Robinson2011-ve}. These plasticity windows have been observed in 
several organisms and brain regions; a typical window has a time constant on the order of 10s 
of milliseconds but have been observed to be on the order of 10s of minutes 
\cite{Froemke2002-be, Zhang2000-lb, Allen2003-rw, Markram1997-ln, Citri2008-kv}. The 
average synaptic dynamics are given by averaging over a time-window which is longer than 
the time-scale of the plasticity window and of the inverse frequencies of the forcing and the 
response stimuli but shorter than any long term plasticity changes:
\begin{align}
	\tau'\frac{dS(x,y,T)}{dT} = \int_{-\infty}^{\infty} \langle U(x,T+s)H(s)A(y,s) \rangle ds, \label{eq:robinson}
\end{align}
where $U=Q(u), A=Q(a)$ (the firing rates of the post-synaptic and pre-synaptic populations respectively), $\langle \cdot \rangle$ denotes averaging, and $\tau^\prime$ is the time-scale of synaptic dynamics. In the case of no electrical activity present in the pre-synaptic field there will be a constant level of spontaneous firing $A$ inducing an electrical activity and firing rate $U$ which in turn will lead to run-away synaptic dynamics. We are interested in the dynamics of the average rate of synaptic change and the expected synaptic values; in later sections this will be taken as an adiabatic expansion and averaging over stimulus input locations. Therefore, the above equation should include a noise term $\eta$, which we shall take to have a strength $\kappa$, to incorporate the small deviations of spontaneous activity:
\begin{align}
	\tau^\prime\frac{dS(x,y,T)}{dT} = \int_{-\infty}^{\infty} \langle U(x,T+s)H(s)A(y,s) \rangle ds + \kappa\eta(x,y,t).
\end{align}
\paragraph{Regularisation}
Several regularisation rules have been posed to stabilise these unstable Hebbian dynamics and are broadly classified in the form of subtractive and multiplicative rules \cite{Abbott2000-gl}. In this study we choose a subtractive normalisation rule to stabilise the dynamics, assuming there is some atrophic factor to regulate the unbounded growth of synapses, governed by parameter $\lambda$, released at each location:
\begin{align}
	\begin{split}
		\tau^\prime&\frac{dS(x,y,T)}{dT} + \lambda S(x,y,T)=  \int_{-\infty}^{\infty} \langle U(x,T+s)H(s)A(y,s) \rangle ds + \kappa\eta(x,y,t).
	\end{split}
	\label{eq:stdpfinal}
\end{align}
The idea of a synaptic decay on the basis of metabolic demands has also been introduced in 
a study of organisational behaviour in V1 \cite{Wright2013-td}. We shall study the dynamics 
of  \eqref{eq:stdpfinal} for the remainder of this text.
\paragraph{Perturbations}
We shall assume that in the absence of forcing activity that the post-synaptic field relaxes to a constant solution i.e. there is a constant level of spontaneous firing in the pre-synaptic and post-synaptic fields; note that this is not necessarily the case \cite{Coombes2005-mt}. We then assume that all activity dynamics are small perturbations from these constant rates. Furthermore, if one makes the assumption that the firing rates can be expressed as perturbations from a baseline firing rate,  $U(x,t) = U_0 + \delta U(x,t)$ and $A(x,t) = A_0 + \delta A(x,t)$, then taking Fourier transforms the average change in plasticity in the un-regularised dynamics can be expressed as:
\begin{equation}
	\frac{1}{2\pi}  \int_{-\infty}^{\infty} \delta \hat{U}(x,\omega) \hat{H}(\omega)^* \delta\hat{A}(y,\omega)^*d\omega,
\end{equation}
where $\hat{\cdot}$ denotes the Fourier transform, and $\cdot^*$ denotes complex conjugation \cite{Robinson2011-ve}. 

\paragraph{Input Stimulus}
We shall consider two classes of input stimulus: mono-directional and bi-directional (radial) 
waves. Mono-directional waves propagate either to the left/right at speed $c$ starting at 
some time $t_0$ and some starting position $y_0$ finally finishing at some time $t_1$. We 
note that these terms are rooted in a two dimensional consideration of the problem. A 
mono-directional wave might travel along a single radial angle whilst a radial wave travels 
isotropically. These choices allow for a description of the waves observed in the retina  
\cite{Meister1991-mu, Stafford2009, Ackman2012-uu}. Of the two the mono-directional 
wave is the most appropriate model but the bidirectional wave is an equivalent but 
analytically preferable case as we show in Section \ref{sec:monotoradial}. Letting 
$r(y,t)=y-ct-y_0$, these inputs accordingly take the form:
\begin{equation}
	a(y,t) = (\Theta(t-t_0)-\Theta(t-t_1))h(r(y,t)). \label{wave:plane}
\end{equation}
Radial inputs are similar, simply propagating in both directions:
\begin{align}
	a(y,t) = \frac{(\Theta(t-t_0)-\Theta(t-t_1))(h(r(y,t))+h(r(y,-t)))}{2}. 
	\label{wave:radial}
\end{align}
In both cases $h$ is used to denote the shape of the propagating wave-form.  The choice of 
$h$ is left to be general but can be thought of as a travelling Gaussian wave-packet. A 
function may be approximated by a linear sum of appropriately weighted Gaussian's and so 
this forms a basis set and we will consider the simple case in Section \ref{sec:param}.
\paragraph{Plasticity Windows}
There are two general forms of plasticity considered: time symmetric and time asymmetric plasticity. Time symmetric plasticity, also called Correlation Dependent Plasticity (CDP), means that connections are strengthened by spikes that are separated by short times and weakened by medium-long time separated spikes, but in which the ordering of the spikes is not important. Time asymmetric plasticity, or STDP, means that not only does the temporal closeness of pre-synaptic and post-synaptic spikes matter but the ordering in which they occur: post-synaptic firing that occurs before pre-synaptic firing weakens the connection and vice-versa. A canonical form of these two rules expressed as a plasticity envelope is given by:
\begin{equation}
	H(s) = \begin{cases}
		A_+\exp(-\frac{s}{t_p}) \quad &s\geq 0\\
		A_-\exp(\frac{s}{t_p}) \quad &s <0
	\end{cases}
\end{equation}
where $A_-=A_+$ for CDP and $-A_-=A_+$ for STDP, and $t_p$ is the time-scale of the plasticity \cite{Abbott2000-gl}. The Fourier transforms of these learning rules are:
\begin{align}
	\hat{H}_{CDP}(\omega) &= \frac{2A_+}{1+\omega^2 t_p^2} \label{rule:CDP}\\
	\hat{H}_{STDP}(\omega) &= \frac{2A_+\omega it_p}{1+\omega^2 t_p^2}. \label{rule:STDP}
\end{align}
In summary, a membrane signal is generated in the post-synaptic region on a fast-time scale which is supported by recurrent connections and generated by input from a pre-synaptic region. The spatial-temporal patterns of the pre-synaptic and post-synaptic activity then inform synaptic changes between the two regions on a slow time scale in accordance with a plasticity rule.

\section{Analysis}
We shall make the assumption that our connectivity kernels, pre-synaptic stimuli, and post-synaptic activity and firing rates are elements of Schwartz space i.e. the functions and derivatives that define these rates decay quickly at long range and they are localised. This assumption is made to ensure bio-physical realism. Connectivity kernels typically have short-range and long-range interactions but they do not interact at all with very distal connections and their functions must accordingly decay at infinity. Similarly, due to these recurrent connectivity kernels, electrical signals only seem to be able to support themselves on finite distances and they too must accordingly decay. The assumption of Schwartz functions ensures that we can take Fourier transforms and makes formulating our problem in Fourier space desirable.

\paragraph{Approximating Input Stimulus}
The inputs that we specified earlier are biologically realistic but will become more tractable if we are able to remove one of the Heaviside functions; this would amount to a stimulus propagating to infinity after being initialized. To show this we need to demonstrate that the synaptic change induced by this different stimulus is arbitrarily small when compared to the synaptic change induced by the true stimulus. This is realised by the rapid decay of the plasticity window and shown formally in Lemma \ref{lemma:decay}; see Appendix A.

\paragraph{Activity Dynamics}
It was reasoned on physical grounds that in the absence of pre-synaptic stimulation the only 
post-synaptic solution is a static, constant level of activity; we are interested in calculating 
perturbations away from these baseline levels.  We shall assume the baseline is sufficiently 
close to the origin that the logistic function is analytical and has a convergent Taylor series 
expanded around $u = u_0$.  Therefore, a good approximation is:
\begin{equation}
	Q(u) = Q(u_0) + u Q'(u_0). \label{eq:firingapproximation}
\end{equation}
This can then be inserted in \eqref{eq:amari} and a Fourier transform can be taken to yield:
\begin{align}
	\hat{u}(k,\omega) = Q&(u_0)\Gamma(k,\omega)(\hat{W}(k)+\hat{S}(k))\delta(k)\delta(\omega) + Q'(u_0)\hat{S}(k)\Gamma(k,\omega)\hat{a}(k,\omega),
\end{align}
with $\Gamma(k,\omega)=(1 + i\tau\omega - \hat{W}(k))^{-1}$. Now recognising that 
$Q(u_0)\Gamma(\omega,k)(\hat{W}(k)+\hat{S}(k))\delta(k)\delta(\omega)$ corresponds to 
the static solution, i.e. the baseline activity level, we can write an expression for the Fourier 
transform of the perturbation of the activity level: 
\begin{equation}
	\delta\hat{U}(k,\omega) =\hat{S}(k)\Gamma(k,\omega)\hat{a}(k,\omega),
\end{equation}
where $U(x,t)=Q(u(x,t))$. The Fourier Transform of the perturbation from the baseline rate in 
the pre-synaptic field, $\delta A(y,t)$ is trivial to compute: $ \delta \hat{A}(k,\omega)=\delta 
\hat{a} (k,\omega)$. This is all we need to explicitly compute the synaptic change between 
any two points in the pre-synaptic and post-synaptic field. 
\paragraph{Synaptic Dynamics}
We shall assume that the synaptic field, and synaptic changes are isotropic; $S(x,y,T) = S(x-y,T)$ for all $T$. Then making the approximation of the firing rate, and taking spatial Fourier transforms the synaptic change can be written:
\begin{equation*}
		\delta(p+k)\frac{d\hat{S}(k,T)}{dT} = \delta(p+k)(S_0 - S_1 \hat{S}(k,T)) +  S_2\int_{-\infty}^{\infty}  \hat{a}(\omega,p)\hat{a}(\omega,k)^* \hat{H}(\omega)^* \hat{S}(k,T)^*\Gamma(\omega,k)  d\omega ,
\end{equation*} 
where $S_0$, $S_1$, and $S_2$ have absorbed the time constant, regularisation constants, baseline firing rate, and the Fourier normalisation terms. We have kept the sign of $S_1$ negative to indicate its relationship with the decay constant $\lambda$. Integrating with respect to $p$, the above equation may be solved as:
\begin{align}
	&\frac{d\hat{S}(k,T)}{dT} = S_0 - S_1\hat{S}(k,T) + S_2\hat{S}(k,T)^*\int_{-\infty}^{\infty}\mathcal{B}(\omega)\hat{a}(\omega,k) \hat{H}(\omega)^* \Gamma(\omega,k) d\omega,
\end{align}
where $\mathcal{B}(\omega) =  \int_{-\infty}^{\infty} \hat{a}(\omega,p)^*dp$. The connectivity kernel $S$ in position space is physically required to be real. We can write it as the composition of odd and even functions. Then, from conjugate symmetry it follows that its Fourier transform is then composed of a real part consisting of the linear combination of the Fourier transforms of its even components, and an imaginary part consisting of the linear combination of the Fourier transforms of its odd components. For $S$ to remain real its derivative must have an even function as its real component, and an odd function as its imaginary component. Denoting,  
\begin{equation} 
	G(k)= \int_{-\infty}^{\infty} \left(\int_{-\infty}^{\infty}  \hat{a}(p,\omega)^*dp\right)\hat{a}(\omega,k) \hat{H}(\omega)^* \Gamma(k,\omega) d\omega,
\end{equation}
we can see that if $G(k)$ is even and real, or odd and purely imaginary, then the above equation can be separated into odd and even parts and solved as two independent ODEs. Attention will be restricted to the even form of $G(k)$ as we will show in the next section that this must be the case. Denoting $S_{O}(x,T)$ and $S_E(x,T)$ to be the odd and even parts of the coupling function in position space these ODEs are then: 
\begin{align}
	&\frac{d\hat{S}_O(k,T)}{dT} =-(S_1 + S_2 G(k)) \hat{S}_O(k,T) \\
	&\frac{d\hat{S}_E(k,T)}{dT} = S_0 + (S_2 G(k) - S_1) \hat{S}_E(k,T) .
\end{align}
Therefore, in the asymptotic limit,   provided $S_1 + S_2 G(k)>0$ the odd components of the 
initial organisation decay to zero and provided $S_1 > S_2 G(k)$ the even components have 
solution:
\begin{equation}
	\hat{S}_E(k) = \frac{S_0}{S_1 - S_2G(k)}. \label{eq:asymptoticsynapse}
\end{equation}
The final organisation is therefore dictated by the initial even components and the form of 
$G(k)$. We will show that $G(k) > 0$   which is sufficient to satisfy the above conditions. The 
form of $G$ is prescribed the learning rule employed and the input stimulus used, we shall 
refer to it as the training function.

\paragraph{Mono-Directional Propagation} \label{sec:monotoradial}
If we suppose the input stimulus is $a(y,t) = \Theta(t) h(y-ct-y_0)$ then it is fairly 
straightforward to show that the training function $G$ is not even and therefore will not work, 
for our purposes, as a training function. However, if we assume that the synaptic changes are 
adiabatic or reasonably small and we assume that the proportions of waves propagating left 
and right are equal then the average synaptic dynamics induced by inputs of the 
mono-directional form  \eqref{wave:plane} are the same as the dynamics induced by inputs 
of the radial form  \eqref{wave:radial}. Therefore, we shall continue the analysis for radially 
propagating inputs.
\paragraph{Radial Propagation}
Presume the input stimulus is in the form $a(y,t) = \Theta(t) (h(y-ct-y_0)+h(y+ct-y_0))$. Taking two Fourier transforms yields:
\begin{align} 
	\begin{split}
		\hat{a}(p,\omega)& = \frac{1}{2} e^{-2\pi i y_0 p}\hat{h}(p) \left( \delta (w+cp) + \delta 
		(w-cp) + \frac{2iw}{\pi (w-cp)\pi (w+cp)}\right).
	\end{split}
\end{align}
Then integrating with respect to $p$ by using the Cauchy Residue Theorem and evenness of the last term and $\hat{h}$ gives:
\begin{equation} 
	\label{integrateap}
	\int_{-\infty}^{\infty}\hat{a}(p,\omega)^*dp = \left(1 + \frac{2}{c} \right) \hat{h}\left(\frac{\omega}{c}\right)^* \cosh\left(2\pi i y_0 \frac{\omega}{c}\right).
\end{equation}
$G(k)=G(k; y_0)$, and if we assume that the synaptic changes at each time step are small then the average synaptic change can be written as:
\begin{equation}
	\left \langle \frac{d\hat{S}(k,T)}{dT} \right \rangle = S_0 - S_1\hat{S}(k,T) + S_2 \hat{S}(k,T) \langle G(k; y_0, c) \rangle. 
\end{equation}
The asymptotic limit, which we are ultimately interested in, will approach this average and for 
the remainder of this work we shall drop the angle brackets. Let $g(k;c) = (\hat{H}(ck)^* 
\Gamma(k,ck)+\hat{H}(-ck)^* \Gamma(k,-ck))/c$. Equation \eqref{integrateap} can then be 
inserted into the expression for $G(k)$ and the Dirac-Deltas can be integrated. Then, we 
integrate out $y_0$ by assuming it is distributed over some interval of length $L$ giving 
exponential integral functions which vanish as $y_0 \rightarrow \infty$ yielding:
\begin{equation}
	\frac{d\hat{S}(k,T)}{dT} = S_0 + \hat{S}(k,T)\left(S_2 g(k;c) |\hat{h}(k)|^2 - S_1 \right). 
	\label{eq:finalsyn} 
\end{equation}
Showing that $g(k;c)$ is even may be done by direct substitution for both STDP and CDP 
rules under the assumption that both $W$ and $h$ are even. It then follows that all $G(k)$ 
are even. It remains to be shown that $G(k)$ is restricted to being non-negative or 
non-positive. All the scaling constants are positive and it is therefore clear for the STDP rule 
that $G(k)\geq 0$, while for the CDP rule $G(k)$ is never non-positive and is only 
non-negative if $\hat{W}(k)<1$. It is certainly possible that this is the case, but it is not true 
for common choices of $W$.  $W$ is typically chosen to be in the form of a ``wizards-hat" 
with short-range excitation and long range inhibition which is theoretically grounded and 
observed experimentally \cite{Sirosh1994-zv, Phongphanphanee2014-in}.
\subsection{Computational Analysis and Parameter Estimation}
So far, we have proceeded in a general manner without much reference to the recurrent connections or input stimulus (with the exception of wave-speed $c$) and the parameters and functional forms that characterise them. Here we shall specify explicit choices for both of these and examine the consequences on the organisation via computational means. We shall also try and estimate key parameters which contribute to the width, or arbor size, of the final organisation by means of Markov Chain Monte Carlo (MCMC) applied to wild-type and $\beta2$ knockout data. This estimation allows us to both validate the model and estimate biological quantities which have not yet been experimentally examined.

We choose a Gaussian to describe the wave-form of the input stimulus with amplitude and width (variance) parameters of $\sigma_1$ and $\sigma_2$ respectively and with Fourier Transform $\hat{h}(k) = \sigma_1 \sigma_2\exp(-k^2\sigma_2^2/2)$. We then choose a difference of two Gaussians to describe the recurrent connections: $\hat{W}(k) =  r_1\exp(-k^2 r_1^2/2) -  R_1r_2\exp(-k^2 r^2)$. The choice ensures that the dimensional requirements for the propagator are satisfied and that $|W(k)|<1$ for a suitable choice of recurrent connection parameters. These choices mean that there are 16 key biological parameters: $u_0, \tau, \tau^\prime, \kappa, \lambda, A_p, t_p, Q_\text{max}, \beta, \theta, c, \sigma_1, \sigma_2, R_1, r_1, r_2$.

\paragraph{Parameter Analysis} \label{sec:param}
Examination of  \eqref{eq:asymptoticsynapse} shows that $S_0$ (or $\kappa/\tau^{\prime})$ 
serves to stabilise the dynamics at the cost of introducing noise - the Fourier spectrum of a 
biologically realistic organisation will decay to a constant i.e. to a baseline level of white 
noise. A tolerable level of system noise is expected and we will assume that this noise can be 
filtered by some means. The denominator dictates the deviations from this noise and noting 
that for both CDP and STDP  $G(k) \rightarrow 0$ and $G(k)>0$ we have that physically 
viable solutions enforce $0<G(k)<S_1/S_2$ and non-viable solutions contain pairs of 
singularities (via evenness of $G$)  where $G(k)>S_1/S_2$ for some $k$. 

We note that an arbitrarily large wave-amplitude $\sigma_1$ can force a singularity in both 
cases and an arbitrarily large $c$ can force a singularity in the STDP case.  From this we can 
deduce for stability in the STDP case that the maximum wave speed is bound by a contour 
inversely proportional to wave-amplitude and vice-versa. Given the likely biological 
restrictions on amplitude this implies that wave speed could be dictated in part by wave 
amplitude. With this in mind we will set $\sigma_1 = 5$mV for the remainder of this work. This 
ensures that there is a baseline distinguishable level of firing when the wave reaches its peak 
amplitude but the neurons are not near a saturated level thus satisfying the assumptions 
required for the approximation in  \eqref{eq:firingapproximation}. We see that $u_0, 
\tau^\prime, \lambda, A_p, f_\text{max}, \beta, \theta$ are absorbed into $S_0, S_1$, and 
$S_2$ and their effects on the dynamics are immediate: they dictate the absolute measurable 
values of the organisation, not the form. We therefore set these parameters according to 
Table \ref{table:params} for the remainder of this work.

\begin{table}[h!]
	\centering
	\begin{tabular}[width=0.5\textwidth]{| l || l | l | l |}
		\hline
		Param. & Value & Units  & Description \\
		\hline
		$u_0$ & -58 & mV & Resting potential \\
		$\tau$  & 0.1 & s & Activity time-scale \\
		$\tau^\prime$  & 100 & s & Synaptic time-scale \\
		$\kappa $  & 0.001 & syn.mm$^{-2}$ & Synapse density \\
		$\lambda $  & 0.001 & $-$ & Decay rate\\
		$A_p $  & 1 & syn.mm$^{-2}$  s & Hebbian rate\\
		$t_p $  & 1.0 & s  & Hebbian time-scale\\
		$Q_\text{max} $  & 1 & s$^{-1}$   & Max firing rate\\
		$\beta$ & 0.26 & mV$^{-1}$ & Rate steepness\\
		$\theta$ & -45 & mV & Rate threshold\\
		$c$ & 0.1 & mm s$^{-1}$ & Wave-speed\\
		$\sigma_1$ & 5 & mV & Wave-amplitude\\
		$\sigma_2$ & 0.1 & mm & Wave-length\\
		$R_1$ & 1.08 & mm &Recurrent amplitude\\
		$r_1$ & 0.129 & mm &Inhibitory length-scale\\
		$r_2$ & 0.136 & mm &Excitatory length-scale\\
		\hline
	\end{tabular}
	\caption{\label{table:params} The choices made for each of the biological parameters used throughout the text, unless otherwise stated. The length scale is chosen to reflect the scale at which NFT typically applies in the brain and the appropriate scale for the mouse SC \cite{Robinson2005-en}. The resting membrane potential, the threshold voltage, and the voltage scale are estimated to be in line with electrophysiological recordings \cite{Shi2018-er}. These parameters should be carefully measured if a specific biological system is to be closely analysed. \label{table:parameters}}
\end{table}

We can see also that for CDP $G(k)$ attains its global maximum at $k=0$ meaning that its stability is determined entirely by the relationship between $S_1$ and $S_2$. Furthermore, with CDP synaptic changes have the potential be to large with no parameter available to mitigate them, in the STDP case the small timescale ensures that the changes are small and the adiabatic assumption is satisfied. We proceed only with the STDP case noting that extending the analysis to a CDP rule would be straightforward but care must be taken in the choice of parameters.

These choices, while considered, have reduced the problem to a single learning rule and several key parameters. We stress that the other parameters must be carefully measured for accurate predictions and are in some sense non-trivial: one can manipulate them biologically and cause a bifurcation in the organisation dynamics. Figure \ref{fig:stablity} demonstrates the manifold in the $c-\sigma_2$ plane for which the model presents plausible (stable) solutions. We have shown only a 2-dimensional slice of the overall manifold for which there are no solutions with singularities, but care should be taken in ensuring that any solution of interest lies within the volume of this manifold for all parameters.
\begin{figure}
	\centering
	\includegraphics[width=0.75\textwidth]{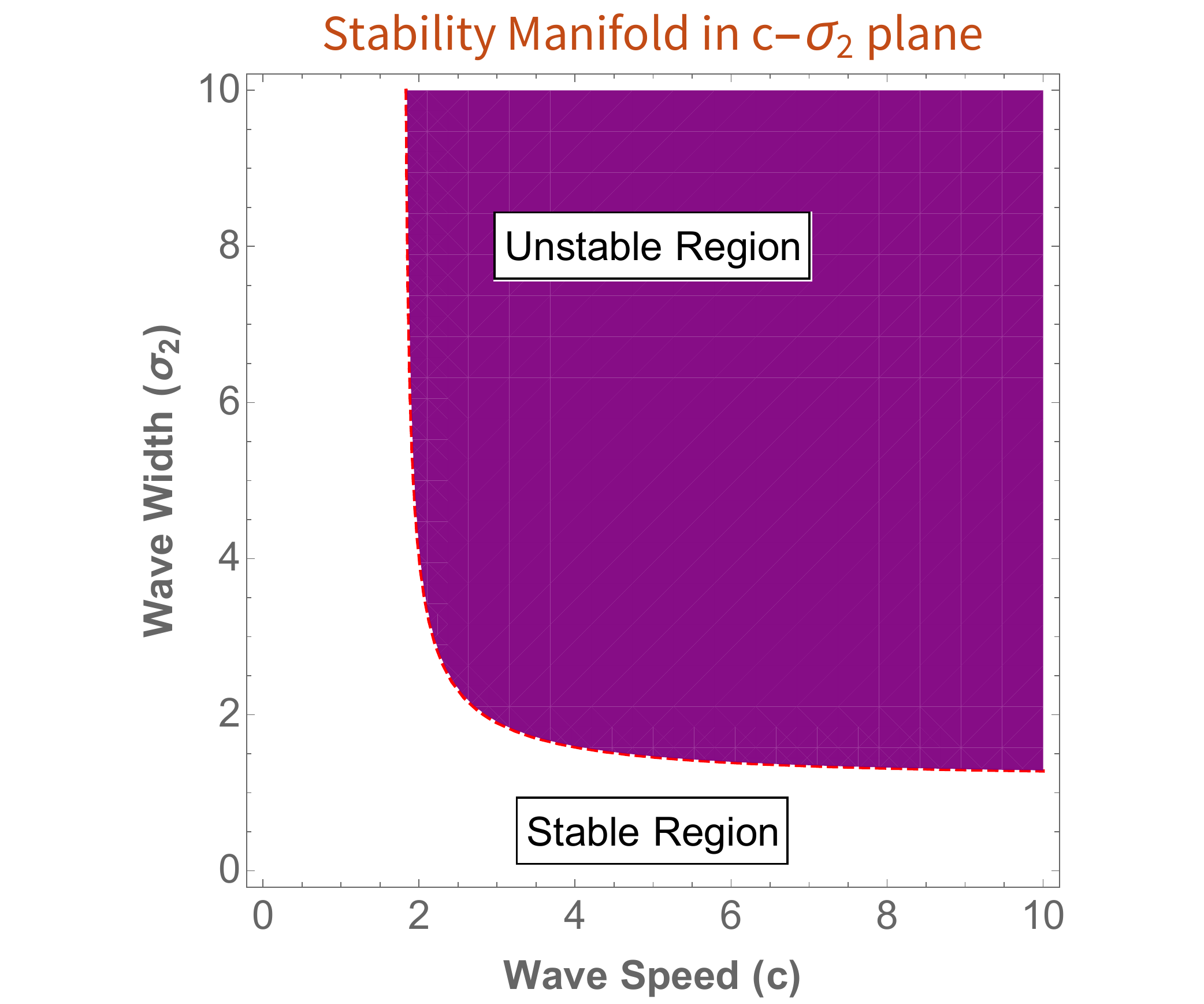}
	\caption{\label{fig:stablity} The manifold in $(c,\sigma_2)$ space which defines the stability of the final organisation. Below the surface solutions do not exhibit singularities and the training function is deemed to be stable; in general, small choices for the parameters exhibit stable synaptic organisations at the cost of arbitrarily small amplitude. The manifold appears to be well-above reasonable estimates for these parameters, ensuring the model is likely stable in plausible biological scenarios.} 
\end{figure}
\paragraph{Fourier Space}
The Fourier transform of $S$ has a characteristic bump near the origin which decays to a 
constant representing a baseline level of noise  i.e. $\hat{S}(k) = c_0 + \hat{\mathcal{S}}(k)$ 
where $\hat{\mathcal{S}}(k)$ is a symmetric function decaying quickly to zero. Note that it is 
possible for $\hat{\mathcal{S}}(k)$ to fall below the noise level which implies that the system 
will be out of phase and suppress signals at this wave length. A typical representation in 
Fourier and real space is shown in Figure \ref{fig:powerdistribution}a. 
\begin{figure}
	\centering
	\begin{subfigure}{0.4\textwidth}
		\centering
		\includegraphics[width=\textwidth]{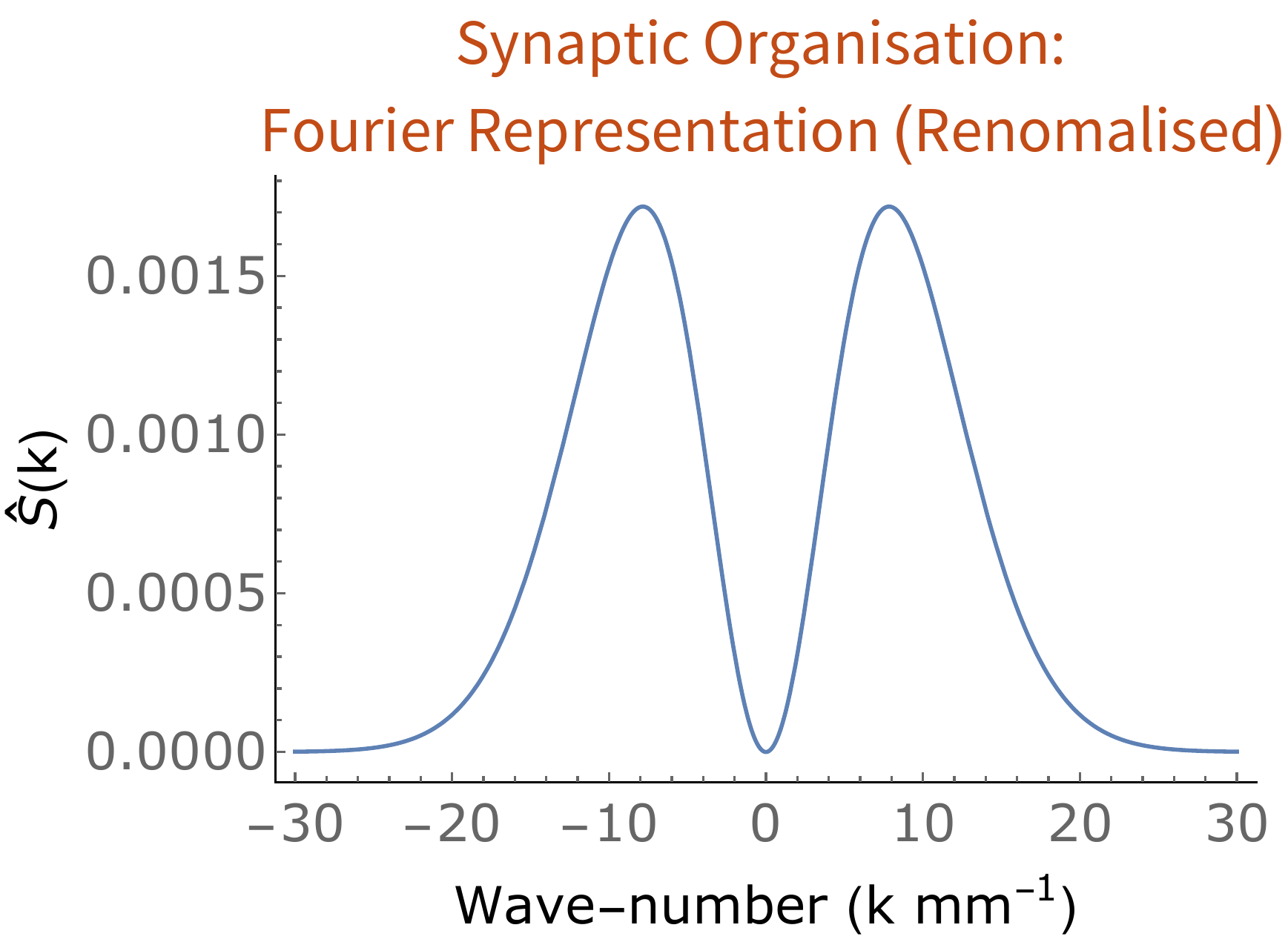}
		\caption{}
	\end{subfigure}
	\begin{subfigure}{0.4\textwidth}
		\centering
		\includegraphics[width=\textwidth]{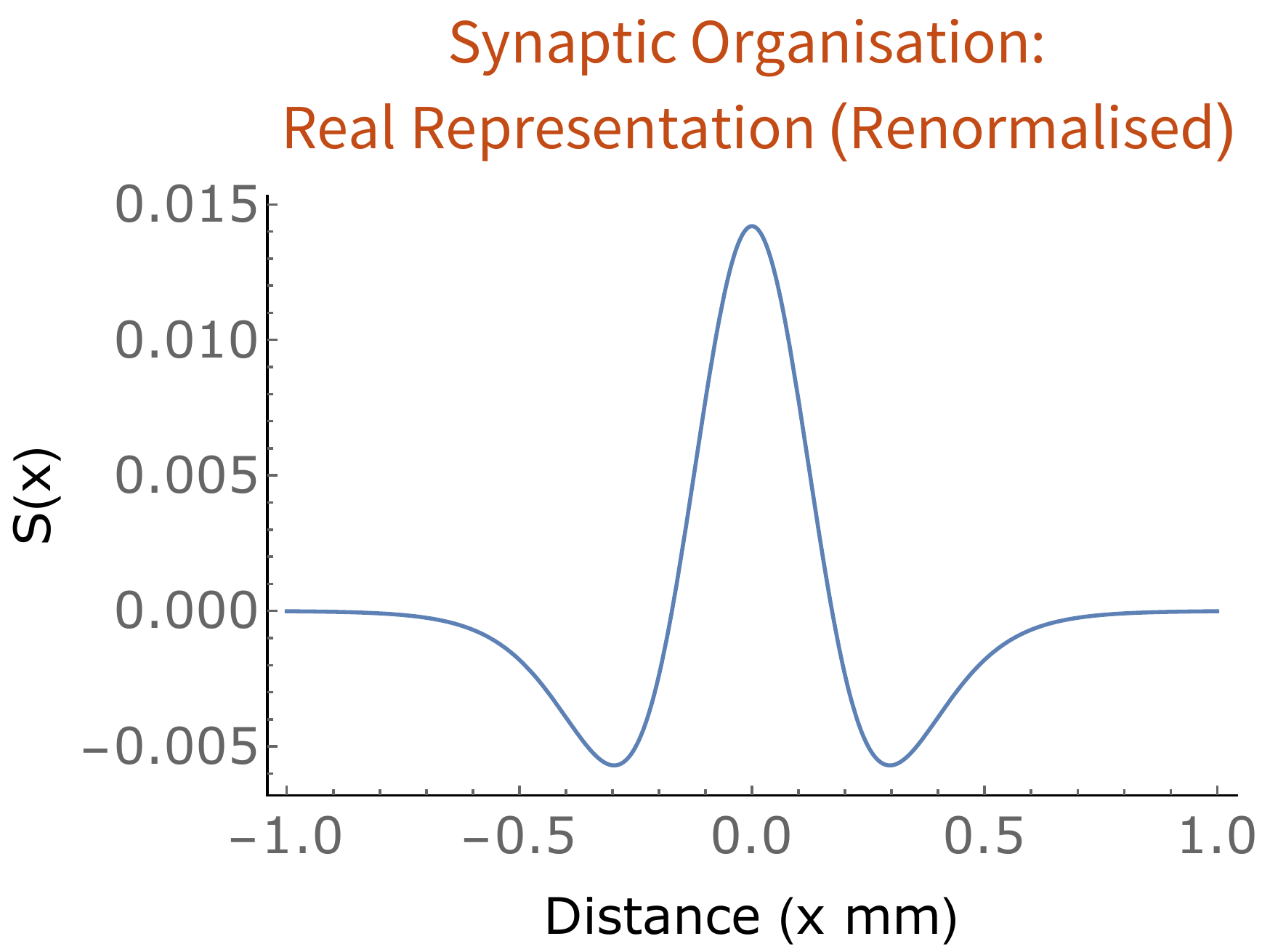}
		\caption{}
	\end{subfigure}
	\caption{\label{fig:powerdistribution} A typical organisation generated with the parameters shown in Table \ref{table:parameters} with (a) showing the representation in Fourier space, and (b) the representation in real space after re-normalisation.} 
\end{figure}

The distribution of the connections in physical space can be found by inverting its Fourier representation which presents a problem with the inclusion of the Dirac-Delta distribution introduced by $c_0$. This problem can be circumvented by realising that the baseline constant representation of all frequencies represents white noise which can be therefore be renormalised and omitted; see Figure \ref{fig:powerdistribution}b. This re-normalisation is done under the assumption that provided the amplitude of $\hat{\mathcal{S}}(k)$ provides a high enough signal-to-noise ratio then this system will be absorbed into already present biological noise which is filtered out in downstream calculations.
\paragraph{Refinement}
The steady state solution of the synaptic distribution $S$ takes its maximum at the origin and 
rapidly decays at large distances. The distributions feed-forward capability is therefore 
dictated by the magnitude at the origin and the rate of the decay. For precise signal 
transmission (or a refined retinotopy) the width of the distribution should be small with 
respect to the length scale. We can estimate width by   taking the inverse of the wave-length 
that maximises the power spectrum:
\begin{equation} \label{eq:width}
	\Omega(\vec{\rho}) =  \frac{1}{\text{argmax}_k\left| \hat{S}(k;\vec{\rho})\right|},
\end{equation}
where $\vec{\rho}$ represents the vector of parameters which define the model. We shall examine the width relationships in the plane of several pairs of variables within a stable region containing no singularities; shown in Figure \ref{fig:parametereximinations}. Refinement tends to decrease in accordance with decreases in $c, \sigma_2, (r_1/r_2), R_1, t_p$, and $\tau$. On the biological scales of interest for the current work the decreases do not appear to be substantial in the $R_1$ and $\tau$ directions. In general the relationships between the variables are non-linear.

\begin{figure*}[h!]
	\centering
	\begin{subfigure}{0.45\textwidth}
		\centering
		\includegraphics[width=\textwidth]{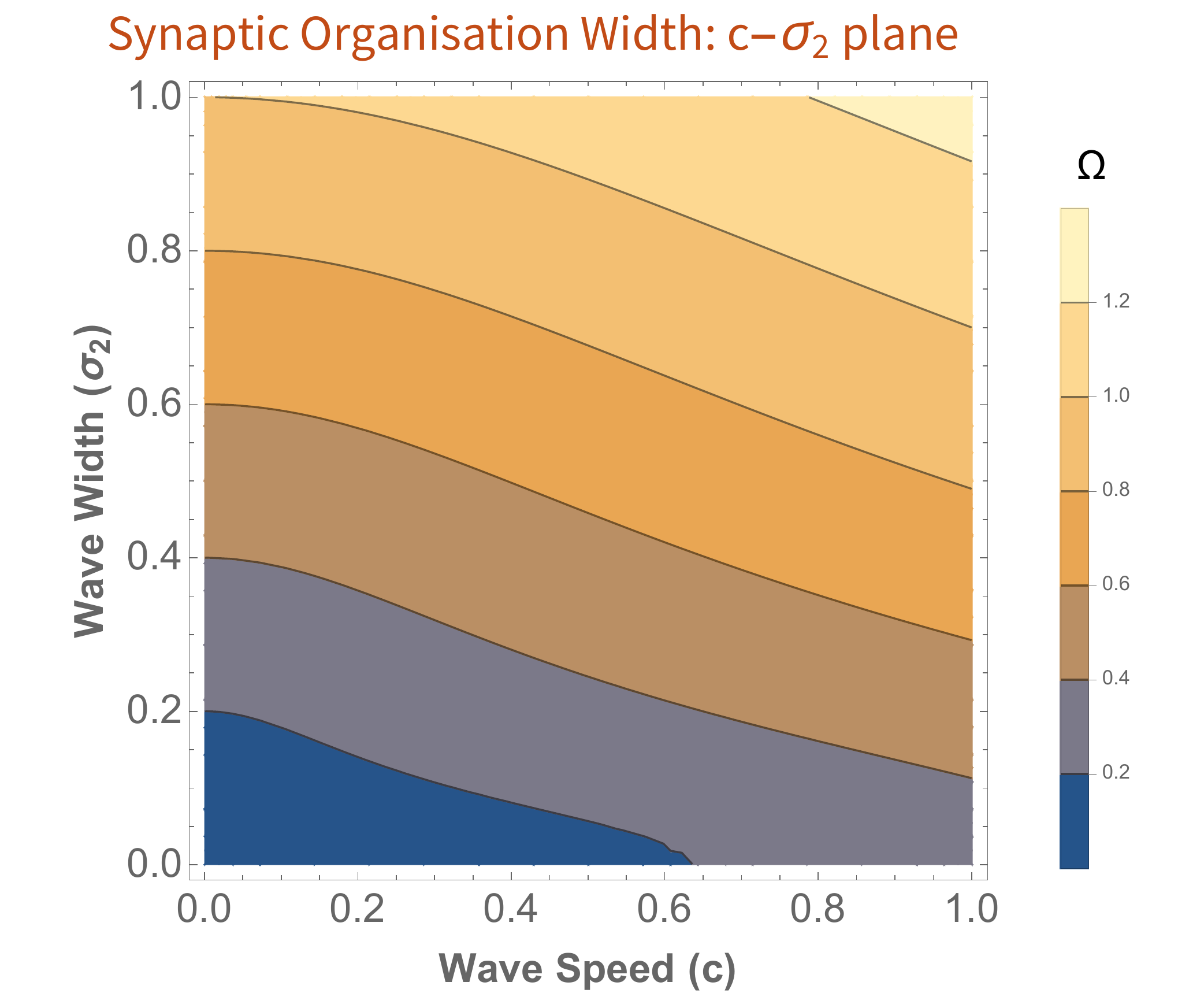}
		\caption{}
	\end{subfigure}
	\begin{subfigure}{0.45\textwidth}
		\centering
		\includegraphics[width=\textwidth]{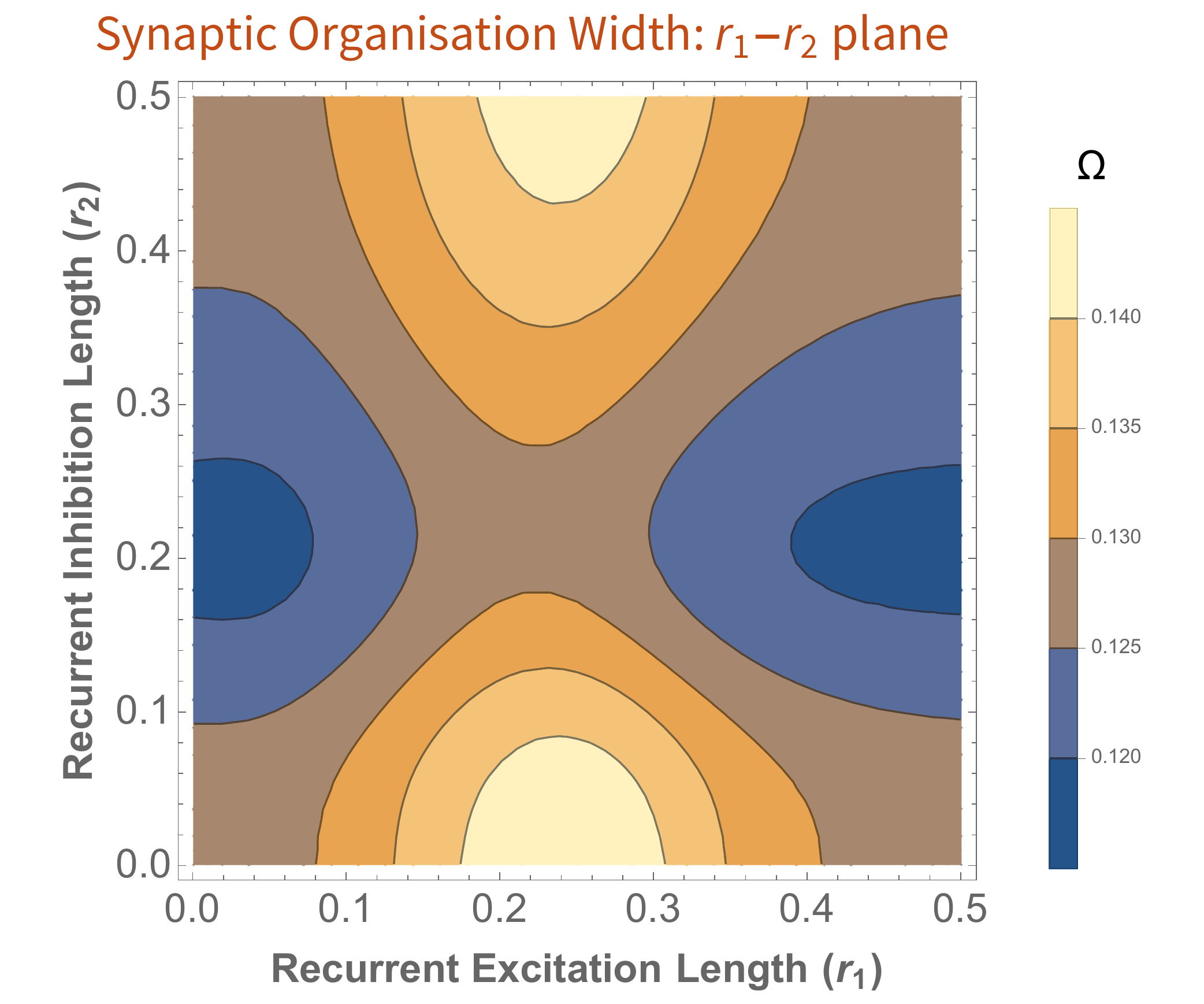}
		\caption{}
	\end{subfigure}
	\begin{subfigure}{0.45\textwidth}
		\centering
		\includegraphics[width=\textwidth]{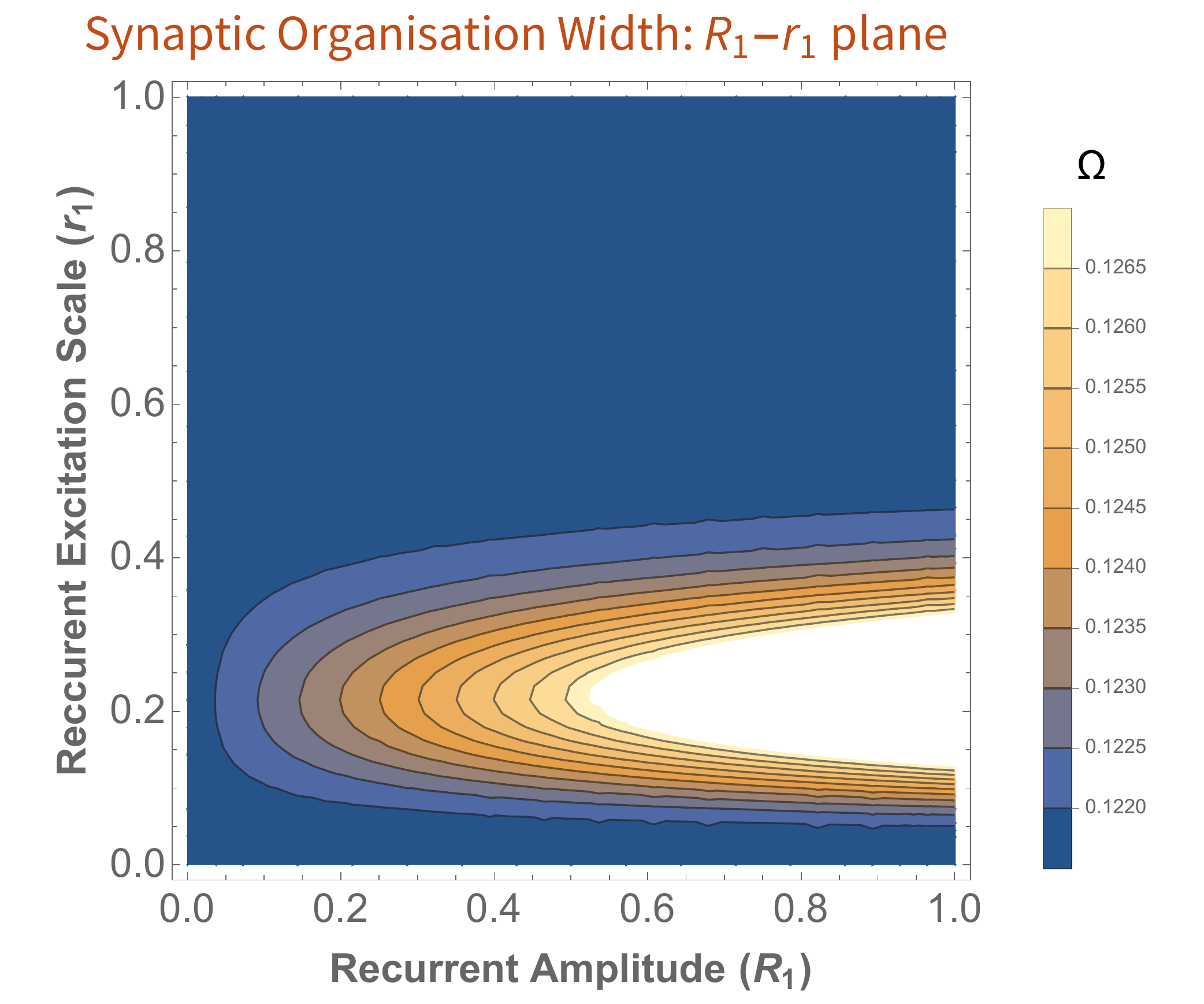}
		\caption{}
	\end{subfigure}
	\begin{subfigure}{0.45\textwidth}
		\centering
		\includegraphics[width=\textwidth]{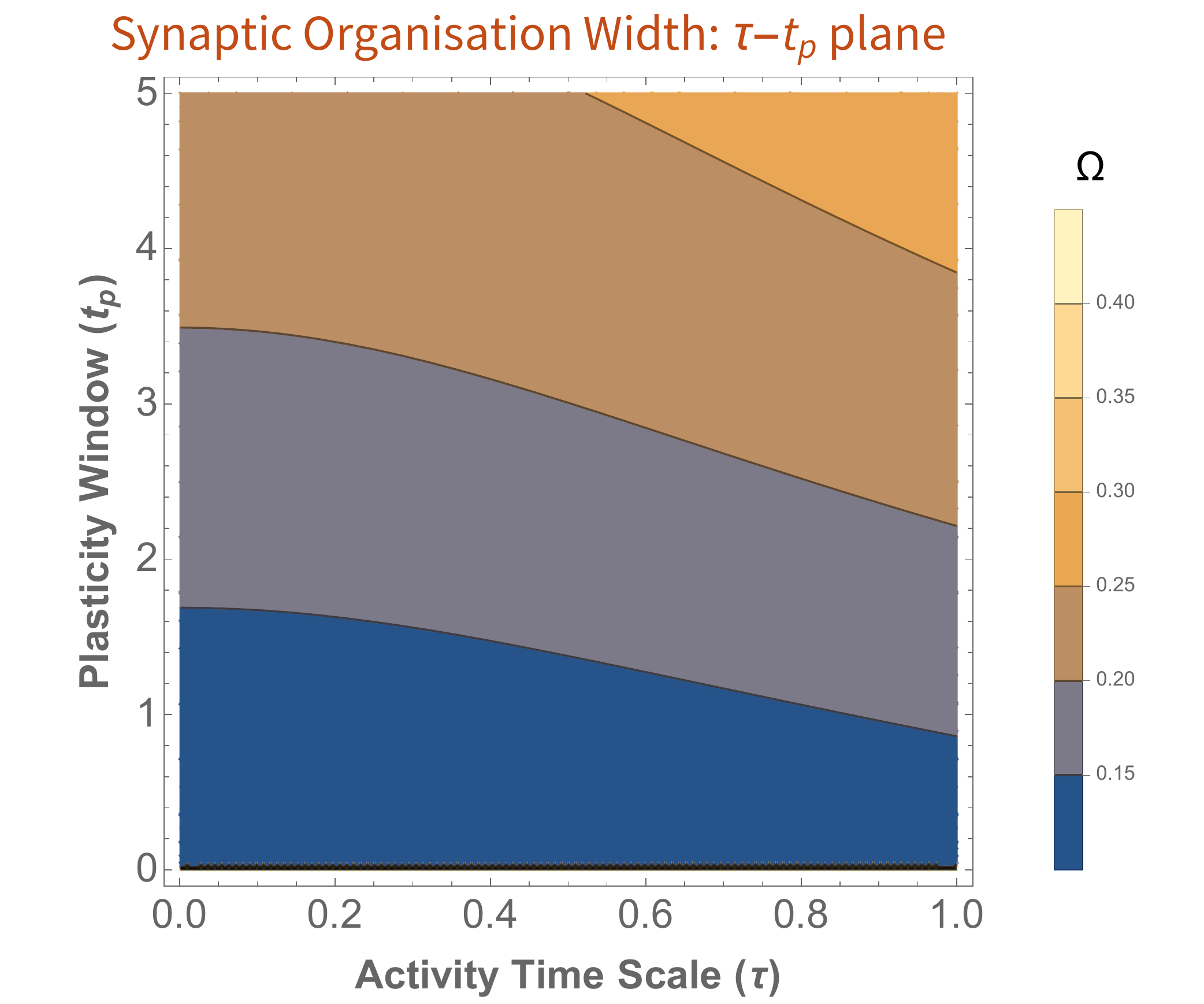}
		\caption{}
	\end{subfigure}
	\caption{\label{fig:parametereximinations} The variation in width ($\Omega$) in four distinct planar slices of the manifold of parameters which influence the models prediction of mean distribution width. Panel (a) shows that width decreases both with wave-speed and wave-width, qualitatively accounting for the differences between the wild-type and $\beta2$ mutant. Panel (b) shows that width decreases with with the ratio of excitation to inhibition in the recurrent connections $W$ suggesting a smaller zone of excitatory support decreases arbor size. There is an anti-symmetry along the line $r_1 = r_2$ which is expected as the dominant connection type switches along this line.  Panel (c) shows that width decreases with recurrent connection amplitude but the effect is not substantial. Panel (d) shows that width predominately decreases in accordance with the plasticity window time-scale, and while the activity time-scale has an effect it is not substantial.}
\end{figure*}

\paragraph{Sensitivity}
In the context of refinement it is prescient to consider which parameters affect the models 
prediction of the width which we have defined. The width given by  \eqref{eq:width} will 
satisfy $dS(k; \vec{p})/dk = 0$ which inserted into \eqref{eq:asymptoticsynapse} yields:
\begin{equation}
	\frac{dg(k; \vec{p_g})}{dk} = 0,
\end{equation}
where $\vec{p_g} = \{c, \tau, t_p, \sigma_2, R_1, r_1, r_2\}$. The width will only vary in 
accordance with these parameters which was confirmed by numerical simulation.

\paragraph{MCMC Parameter Estimation}
The $\beta2$ knock-out in mouse has the effect of altering the spatio-temporal patterns of spontaneous activity in the retina and SC during development \cite{ Stafford2009}. The mutant mice have substantially wider arborisations than in wild-type establishing the importance of activity in refining the retinotopic projection \cite{Dhande2011-jp}. Existing models have not been able to predict this wider arborisation when the patterns of activity associated with the knock-out are replicated in the models mechanisms for activity \cite{Lyngholm2019-fs}.

\begin{figure*}[h!]
	\centering
	\begin{subfigure}{0.475\textwidth}
		\centering
		\includegraphics[width=\textwidth]{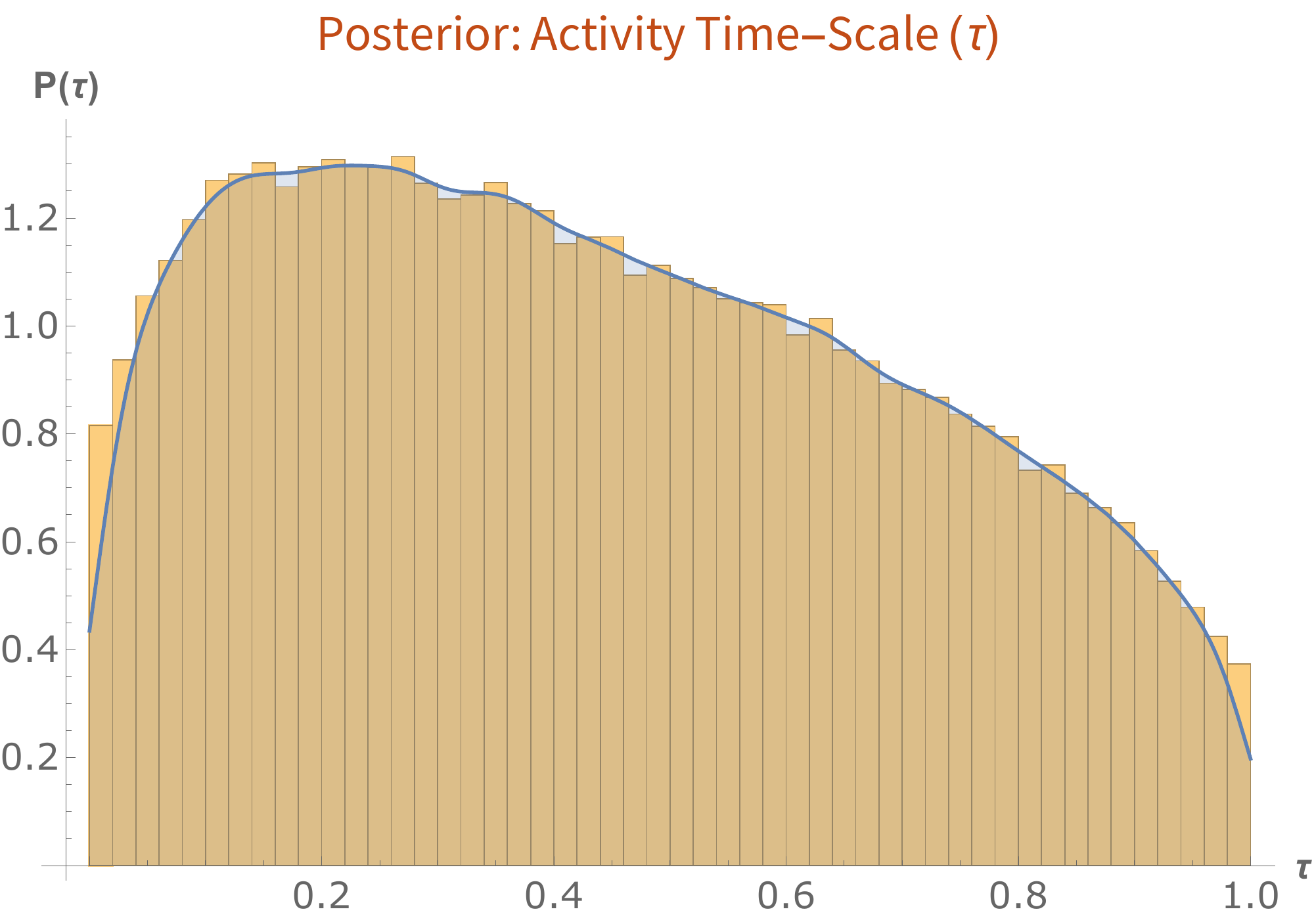}
		\caption{}
	\end{subfigure}
	\begin{subfigure}{0.475\textwidth}
		\centering
		\includegraphics[width=\textwidth]{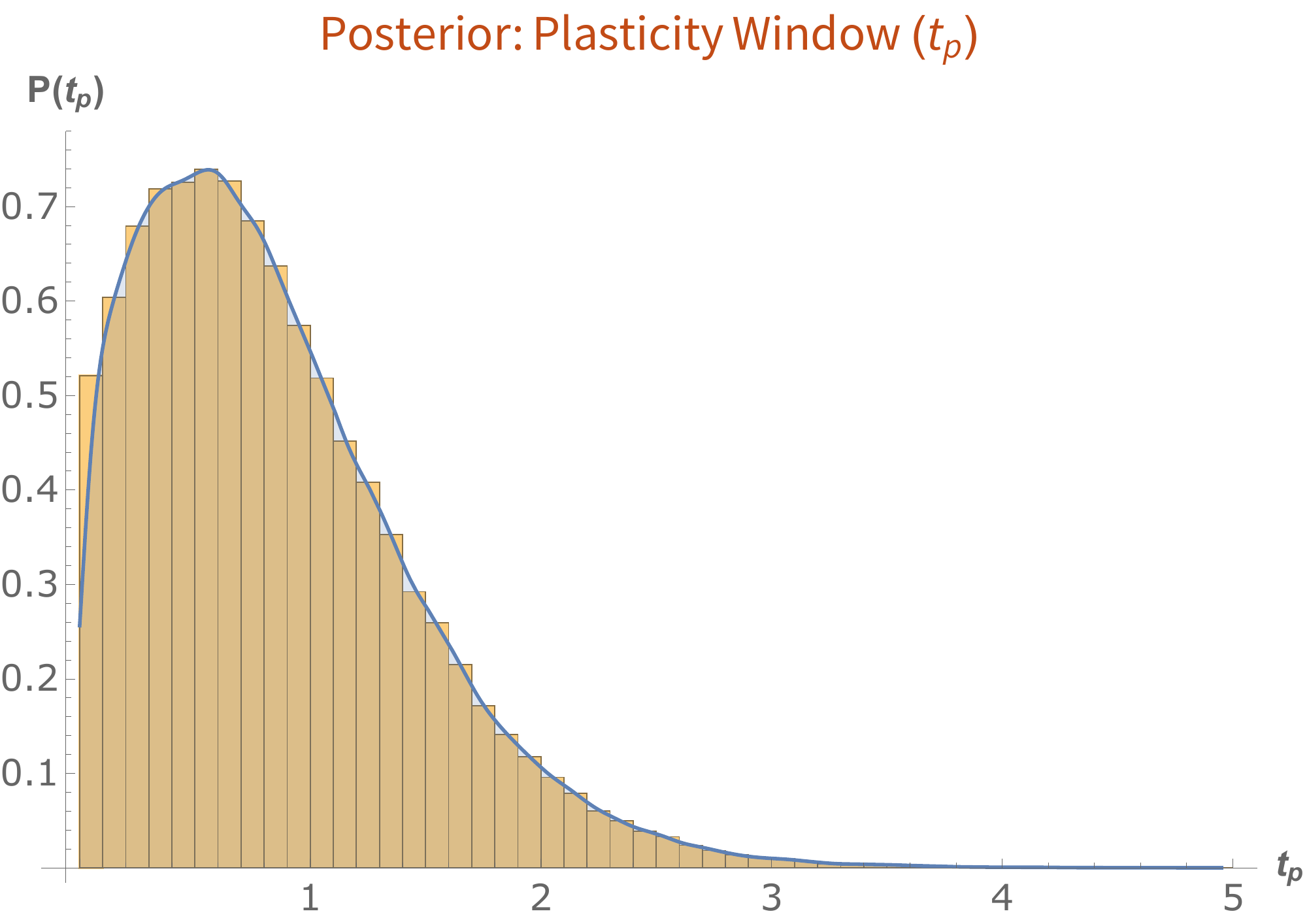}
		\caption{}
	\end{subfigure}
	\begin{subfigure}{0.475\textwidth}
		\centering
		\includegraphics[width=\textwidth]{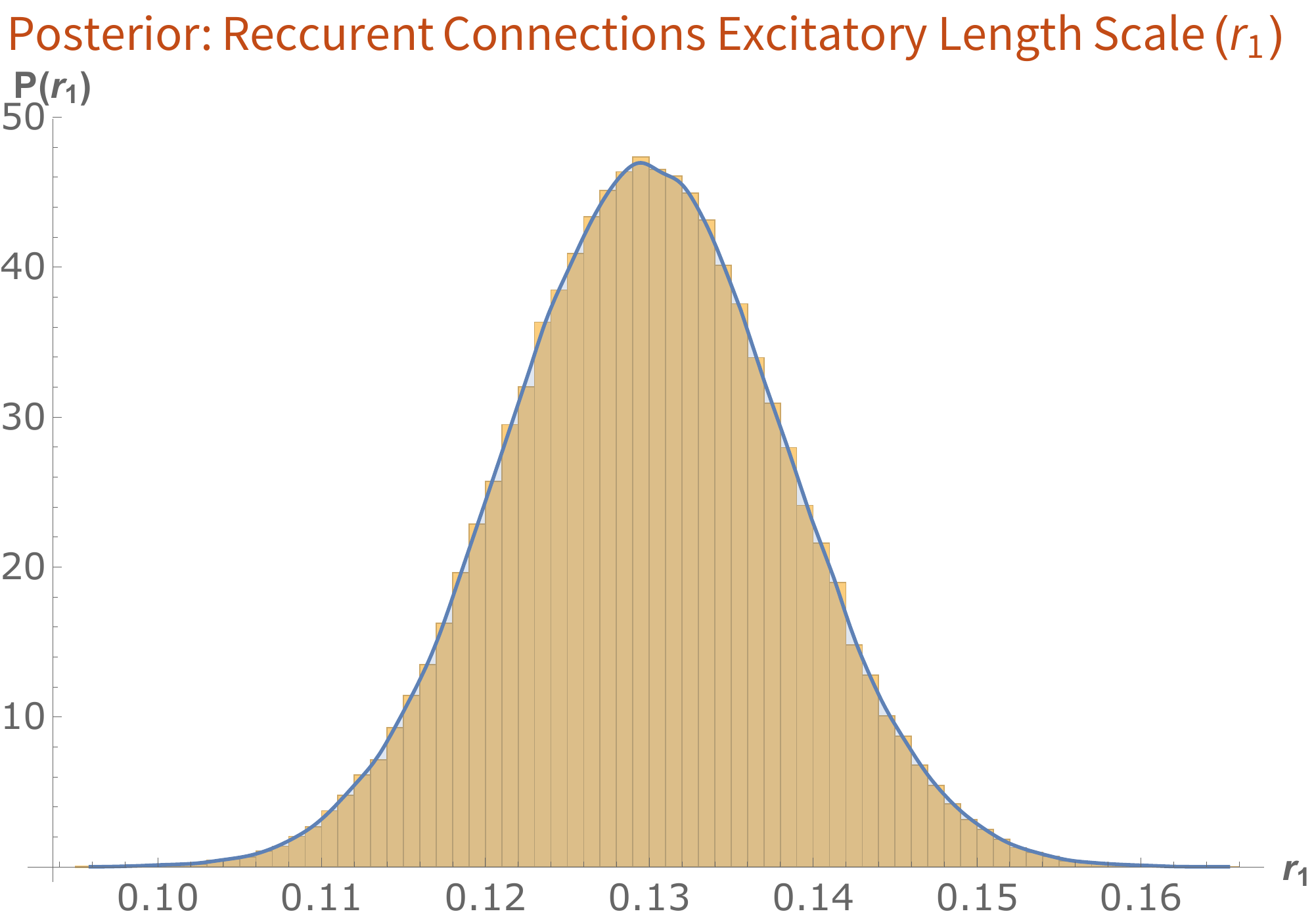}
		\caption{}
	\end{subfigure}
	\begin{subfigure}{0.475\textwidth}
		\centering
		\includegraphics[width=\textwidth]{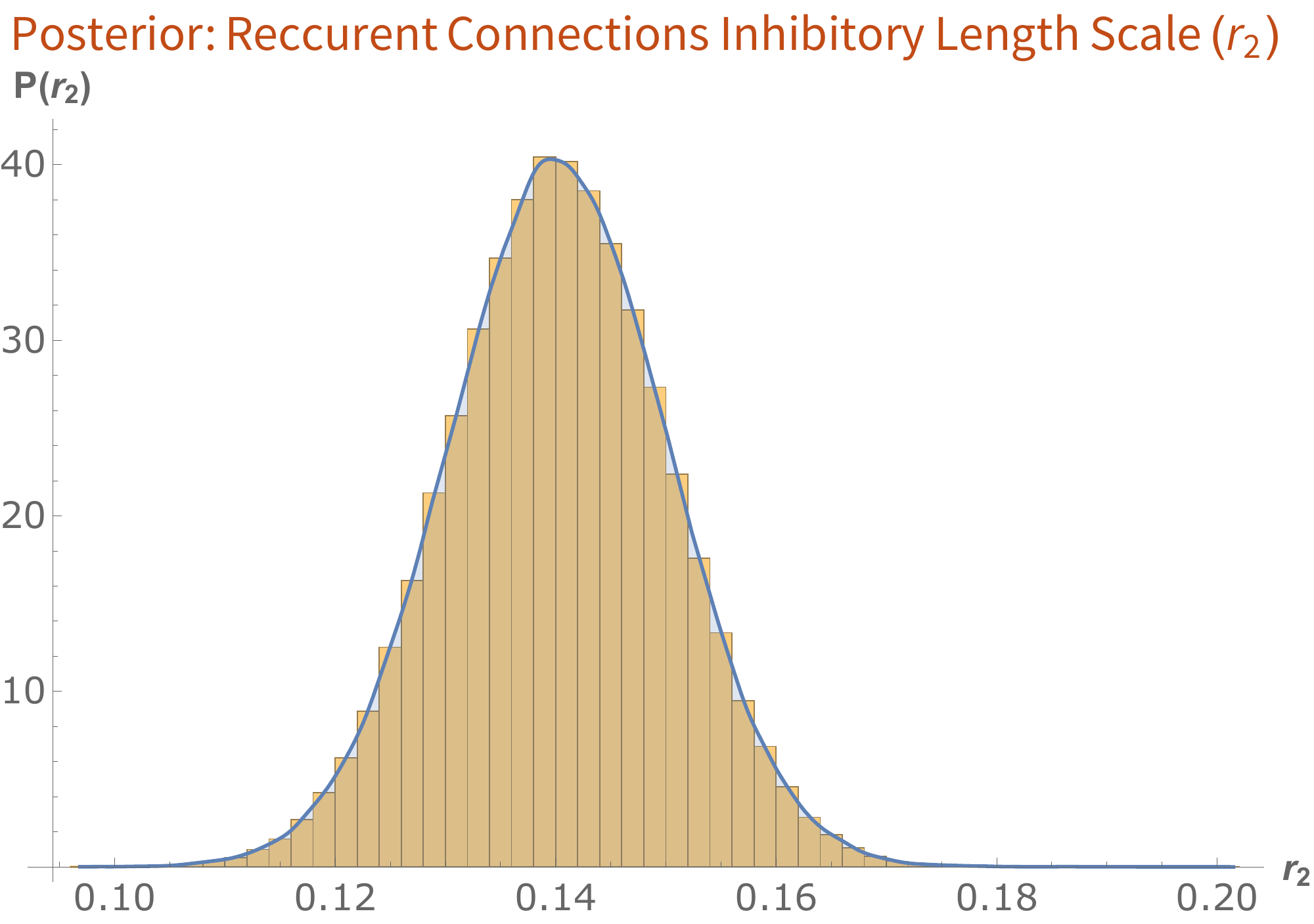}
		\caption{}
	\end{subfigure}
	\begin{subfigure}{0.475\textwidth}
		\centering
		\includegraphics[width=\textwidth]{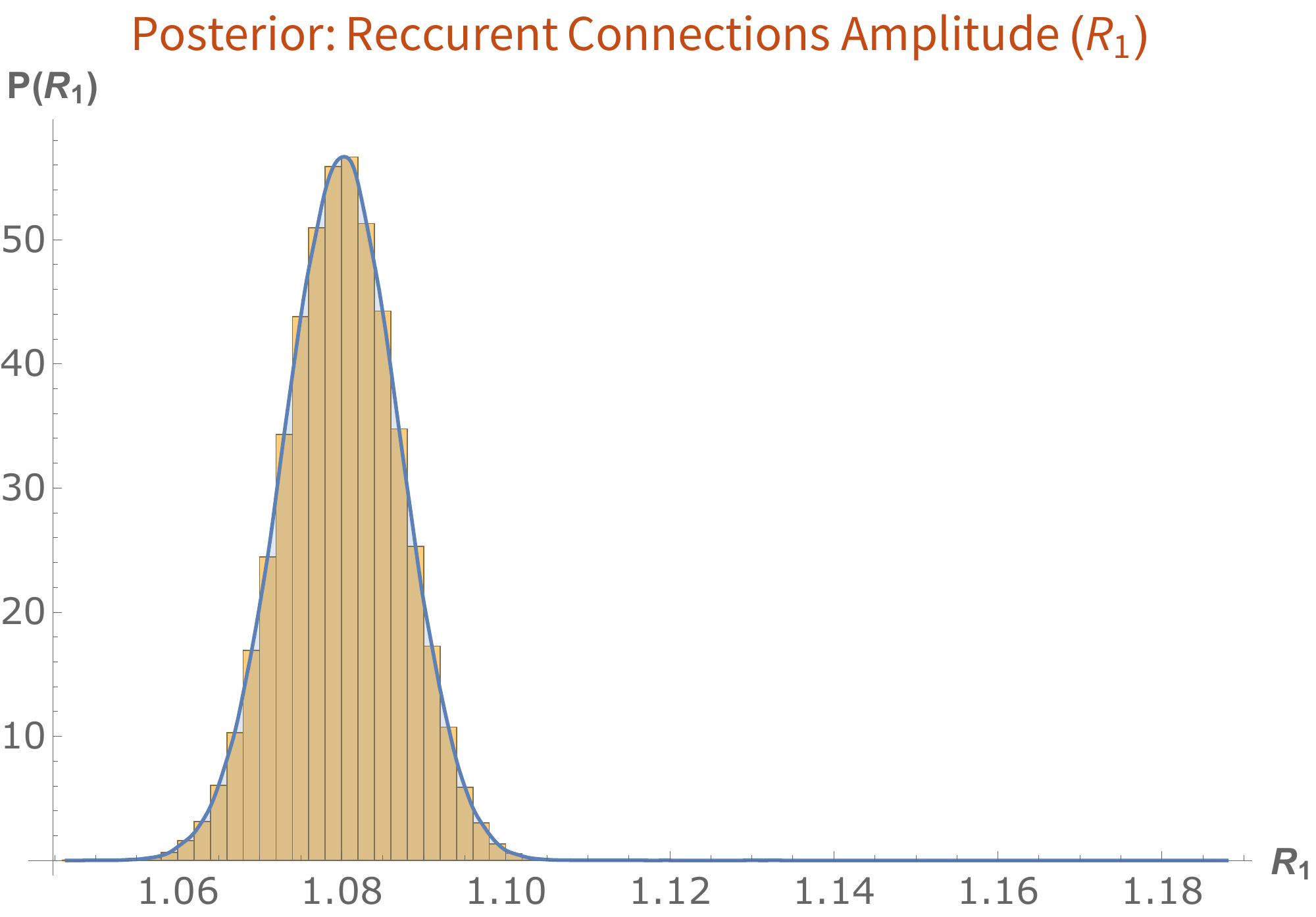}
		\caption{}
	\end{subfigure}
	\caption{\label{fig:posteriors} Panel (a) shows the posterior histogram for the time-scale of activity which is broadly distributed through the search space of [0,1]s but biased towards the lower bound. This broad distribution is concordant with the observation that the time-scale of activity induces relatively small variations in the organisation width; see Figure \ref{fig:parametereximinations}. Panel (b) shows the posterior histogram for the time scale of the plasticity window which is maximised around 0.6s. The posterior histograms for the recurrent connection parameters $(r1, r2, R1)$ are shown in Panels (c-e) and are tightly constrained by their informative priors suggesting that there is no predicted effect on these connections in the $\beta2$ mutant. For all histograms presented an empirical distribution curve was fitted and overlain in blue.}
\end{figure*}

We estimate the arborisation widths as $0.24 \pm 0.077 \text{mm}$ (wild-type) and $0.48 
\pm 0.15 \text{mm}$ ($\beta2$) by taking half the square root of the arborisation area 
reported by \cite{Dhande2011-jp}. We estimate the wave speeds as $0.13 \pm 0.015 
\text{mm s}^{-1}$ (wild-type) and $0.17 \pm  0.03 \text{mm s}^{-1}$ ($\beta2$), and the 
wave-widths as $0.11 \pm 0.012 \text{mm}$ (wild-type) and $0.20 \pm 0.012 \text{mm}$ 
($\beta2$) by taking half the total width reported by Stafford et al  \cite{Stafford2009}. We 
estimate the inhibitory and excitatory lengths scales, and amplitude of the recurrent 
connections to be $0.14 \pm 0.014 \text{mm}$ and $0.13 \pm 0.013 \text{mm}$, and $1.08 
\pm 0.01 \text{mm}$ respectively using the data reported by 
\cite{Phongphanphanee2014-in}. We take our priors on these parameters to be normal 
distributions centred on the estimates with standard deviation corresponding to the 
measurement error. We take uninformative priors on the time-scale parameters assigning 
uniform distributions on [0,1] and [0,10] for the activity time scale ($\tau$) and the plasticity 
window scale ($t_p$), respectively. The MCMC was completed using a dedicated 
Mathematica package \cite{Burkart2017}. The MCMC completed in $10^5$ iterations using 6 
chains with each parameter initialised within 10\% of the mean of its prior. The maximum 
Gelman-Rubin statistic for convergence was $1.00037$ indicating that the chains had 
converged \cite{Gelman1992-pk}. The posteriors for each parameter are reported in Figure 
\ref{fig:posteriors}. The posteriors for the recurrent connections parameters, $r1, r2$, and 
$R_1$ remained tightly constrained by their priors, indicating that the prior estimates were 
well informed and in agreement with the model. The activity time scale is broadly distributed 
throughout the range [0,1]s with a bias towards 0. The plasticity time-scale is distributed 
around a maximum of $0.56$s. The computed $R^2$ statistic was 0.81.
\newpage
\section{Discussion}
If the model is sound and the biological system is allowed sufficient time to reach a reasonable approximation of the asymptotic state then these results suggest that the computational/synaptic structures developed are primarily a result of activity dynamics. Under this model the chemotactic and competitive mechanisms serve to initialise a coarse isotropic retinotopy from which the activity dynamics can refine and ultimately dictate final synaptic organisation. This interpretation augments the understanding of the establishment of retinotopy by suggesting that the final synaptic organisation can be understood in a large part by understanding the spatio-temporal nature of the input stimulus, the recurrent connectivity, and the learning rule. Should the biological system not employ the learning routine until asymptotic stability then the model will still be able to make predictions about the final organisation given precise enough measurements of the relevant parameters. In both instances the model gives testable hypotheses the former of which has been benchmarked against the mouse wild-type and $\beta2$ knock-out mutant.

\paragraph{Organisation}
We have shown that the key aspects of the final organisational structure are dictated by the interplay between the spatio-temporal characteristics of the input stimulus and the structure of the recurrent connections. These dependencies on recurrent connections and input are in accordance with previous analysis performed with a simple Hebbian rule and static input \cite{Takeuchi1979-zy}; the model proposed here, however, allows for richer construction in terms of specifying the input and connections by realising full temporal and spatial dynamics, and more complex structure in the final organisation. We have introduced regularisation rules which allow this organisation to take non-trivial structure when supplemented by system noise which we have assumed is able to be renormalised in downstream biological calculations or via some other mechanism. The regularisation necessitates neurotrophic factors being expressed during development. Finally, the measurable aspects of the organisation are dictated by the precise realisation of the relevant biological parameters.
\paragraph{Refinement}
The results indicate parameter dependence on wave-speed, wave-width, plasticity time-scales, and the ratio of excitation to inhibition widths in the recurrent connections. Principally, parameter changes that would lead to a tighter correlation structure such as smaller wave-widths, slower wave-speeds, and smaller excitatory zones lead to a smaller width of topographic refinement. Interestingly, the time-scale of the plasticity rule has an effect of the width of the final organisation. The $\beta2$ knock-out provides a phenomenological test of this component of model. The knock-out exhibits fast-moving, and hyper-correlated, retinal waves which lead to an imprecise topographic mapping - an effect that has not been captured in existing models. Our model suggests that an increase in wave-speed or wave width will lead to a less-refined map reproducing the results of the knock-out \textit{in silico}; see Figure \ref{fig:parametereximinations}. 

An MCMC parameter estimation was performed using known errors-in-measurement of 
wave-speed, wave-width, and organisation width in wild-type and the $\beta2$ mutant. The 
model predicts the expected mean width of both wild-type and the $\beta2$ knock-out 
within standard error when parametrized by likelihood maximising parameters and provides a 
good explanation of the variance between the wild-type and mutant ($R^2$ = $0.81$).  We 
found the model to be insensitive to the time-scale of activity with the posterior assuming a 
broad posterior over $[0,1]$s with a slight bias towards lower values suggesting that the 
activity time-scale does not account for much of the variance in organisation width. The 
posteriors of the parameters of the recurrent connections were largely dictated by their 
priors suggesting that the priors estimated from available are informative and that the 
$\beta2$ knock-out does not have a substantial effect on the recurrent connections, as 
expected. We do not expect the time-scale of the plasticity window to be affected by the 
knock-out and thus the MCMC allows us to estimate this parameter on the order of seconds. 
The timescale of the plasticity window in two closely related biological systems, the Xenopus 
retinotectal projection and rat visual cortex, are estimated to be on the order of $10^{-2}$ 
seconds \cite{Froemke2002-be, Zhang2000-lb}. Plasticity windows can have significantly 
longer time-scales on the order of 10s of minutes \cite{Citri2008-kv}. Our estimate is notably 
higher than what has been observed in similar systems but is in agreement with the typical 
duration of a wave of spontaneous activity in the developing retina in mice \cite{Xu2015-uc}. 
We might expect a deviation as we are analysing a different biological system. This result 
suggests that the plasticity windows in this system are calibrated to integrate all information 
contained in a spontaneous wave event.
\paragraph{Future Directions}
The analysis presented here has made simplifying assumptions about the statistical 
properties of spontaneously generated waves: these assumptions cannot be expected to 
hold in general. The analysis was also restricted to one dimension: the two-dimensional case 
has a much richer topology and is more relevant as the topographic projection is typically 
organised as a sheet. The analysis can be trivially extended into the plane by using the same 
assumption: every wave-direction is equiprobable. More realistic wave-statistics can be 
simulated numerically and examining the properties of the synaptic distribution generated by 
the data of spontaneous activity in mouse is a future research direction; for example using 
the model of activity proposed by Godfrey and Eglen \cite{Ackman2012-uu, Demas2003-xf, 
Godfrey2009-rs}.

The model predicts that the time-scale of the plasticity window in developing mouse SC neurons is 1-2 orders of magnitude higher than the scale typically used to describe neuronal plasticity in analogous systems. While we do not claim that this prediction represents a ground truth, the model makes several simplifying assumptions and estimations, it is a good candidate for experimental falsification.
\paragraph{Conclusion}
We have developed a modelling framework in which the effects of rich spatio-temporal patterns of activity on topographic refinement can be analysed alongside system specific measurements of parameters. The model posits that the final synaptic organisation is dictated in a large part by the characteristics of this activity suggesting a more involved role for activity, spontaneous or otherwise, in the developing visual system. The model explains topographic defects observed in the $\beta2$ mutant which has had its spontaneous activity patterns altered and on the basis of the mutant and wild-type offers a prediction of the time-scale on which Hebbian refinement operates in mouse development.

\paragraph{Generating Code}
The code used to perform the analysis and generate the images in this project may be found at \url{https://github.com/Nick-Gale/Neural_Field_Theory_Topopgraphic_Development}.

\section*{Appendix: A}
\begin{lemma}
	\label{lemma:decay}
	The synaptic change $\frac{dS_p}{dT}$ induced by a given input stimulus $A_p$ which 
	terminates at some arbitrary $t_1$ can be well approximated by a similar input stimulus 
	$A$ that terminates at $t=\infty$ i.e. $|\frac{dS_p}{dT} - \frac{dS}{dT}|<\epsilon$ for 
	$\epsilon \ll 1$. 
\end{lemma}

\begin{proof}
	Consider a function $A(y,t)$ which propagates to infinity and induces and activity in the 
	post-synaptic field of $U(x,t)$. For physical reasons this function must decay rapidly at 
	infinity implying for all real $t_j$:
	\begin{equation}
		\int_{t_j}^{\infty} A(y,t)dt =\epsilon_j.
	\end{equation}
	Then, due to the rapid decay of the of plasticity function we also have that for all physical 
	realisations of $u$ and for all $t$:
	\begin{equation}
		\int_{-\infty}^{\infty} H(\tau) U_i(t+\tau)d\tau = \xi_i <\infty.
	\end{equation}
	Then, consider the functions $A(y,t)=\Theta(t) h(y,t)$ and 
	$A_p(y,t)=(\Theta(t)-\Theta(t-t_1))h(y,t)$, and the functions $U(x,t)$ and $U_p(x,t)$ which 
	are induced activities from stimulus $A$ and $A_p$. Observe that as a result of the rapidly 
	decaying plasticity window there exists some $\xi$ such that:
	\begin{equation} 
		\left|\int_{t_1}^{\infty} H(\tau) (U_p(x,t)-U(x,t)) d\tau \right|< \xi,
	\end{equation} and: 
	\begin{equation}
		\left|\int_{-\infty}^{0} H(\tau) (u_p(x,t)-u(x,t)) d\tau \right|< \xi,
	\end{equation}
	for all $x$ and $t$. Also, observe that in the limit $t_1 \rightarrow \infty$, $\xi$ tends to 
	zero. Now let $\epsilon_2 = \xi \int_{0}^{t_1} A(y,t)dt$ and note that in the limit $t_1 
	\rightarrow \infty$ this $\epsilon_2$ will also tend to zero, as the integral of $A(y,t)$ is 
	bounded. Finally, suppose $\int_{t_1}^{\infty} A(y,t) dt < \epsilon_1$. Then, $\varepsilon = 
	K_0 (\xi \epsilon_1 + 2\epsilon_2)$ may be made arbitrarily small for sufficiently large $t_1$. 
	Now consider the synaptic change induced by the truncated function $A_p$:
	\begin{align*}
		\frac{dS_p(x,y,T)}{dT} & = K_0  \int_{-\infty}^{\infty }A_p(y,t) \int_{-\infty}^{\infty}   
		H(\tau)U_p(x,t+\tau)   d\tau dt\\
		&<  K_0  \int_{0}^{t_1}h(y,t) \left (\int_{-\infty}^{\infty}   H(\tau)U(x,t+\tau) d\tau  + 2\xi 
		\right)  dt\\
		&<  K_0  \int_{-\infty}^{\infty}\Theta(t) h(y,t) \int_{-\infty}^{\infty}   H(\tau)U(x,t+\tau) 
		d\tau dt +K_0 \epsilon_1 \xi+2K_0 \epsilon_2\\
		&<  \frac{dS(x,y,T)}{dT} + \epsilon.
	\end{align*}
	Therefore, it is a sufficiently good approximation to consider the stimulus propagating to 
	infinity, rather than the stimulus truncated at time $t=t_1$ when calculating the synaptic 
	change.
\end{proof}



\begin{thebibliography}{10}
	

	

	
	
	


	
	
	
	
	
		\bibitem{Abbott2000-gl}
	L~F Abbott and S~B Nelson.
	\newblock Synaptic plasticity: taming the beast.
	\newblock {\em Nat. Neurosci.}, 3 Suppl:1178--1183,  2000.
	\url{https://doi.org/10.1038/81453}
	
	\bibitem{Ackman2012-uu}
	James~B Ackman, Timothy~J Burbridge, and Michael~C Crair.
	\newblock Retinal waves coordinate patterned activity throughout the developing
	visual system.
	\newblock {\em Nature}, 490(7419):219--225,  2012.
	\url{https://doi.org/10.1038/nature11529}
	
		\bibitem{Allen2003-rw}
	Cara~B Allen, Tansu Celikel, and Daniel~E Feldman.
	\newblock Long-term depression induced by sensory deprivation during cortical
	map plasticity in vivo.
	\newblock {\em Nat. Neurosci.}, 6(3):291--299,  2003.
	\url{https://doi.org/10.1038/nn1012}
	
	\bibitem{Amari1977-gc}
	S~Amari.
	\newblock Dynamics of pattern formation in lateral-inhibition type neural
	fields.
	\newblock {\em Biol. Cybern.}, 27(2):77--87,  1977.
	\url{https://doi.org/10.1007/BF00337259}
	
	
	
	
	\bibitem{Bansal2000-ts}
	A~Bansal, J~H Singer, B~J Hwang, W~Xu, A~Beaudet, and M~B Feller.
	\newblock Mice lacking specific nicotinic acetylcholine receptor subunits
	exhibit dramatically altered spontaneous activity patterns and reveal a
	limited role for retinal waves in forming {ON} and {OFF} circuits in the
	inner retina.
	\newblock {\em J. Neurosci.}, 20(20):7672--7681,  2000.
	\url{https://doi.org/10.1523/JNEUROSCI.20-20-07672.2000}
	
	
	\bibitem{Bednar2004-xl}
	James~A Bednar, Amol Kelkar, and Risto Miikkulainen.
	\newblock Scaling self-organizing maps to model large cortical networks.
	\newblock {\em Neuroinformatics}, 2(3):275--302, 2004.
	\url{https://doi.org/10.1385/NI:2:3:275}
	
		\bibitem{Bednar2016-lg}
	James~A Bednar and Stuart~P Wilson.
	\newblock Cortical maps.
	\newblock {\em Neuroscientist}, 22(6):604--617,  2016.
	\url{https://doi.org/10.1177/1073858415597645}
	
	
	
	\bibitem{Buonomano1998-wi}
	D~V Buonomano and M~M Merzenich.
	\newblock Cortical plasticity: from synapses to maps.
	\newblock {\em Annu. Rev. Neurosci.}, 21:149--186, 1998.
	\url{https://doi.org/10.1146/annurev.neuro.21.1.149}
	
	\bibitem{Burkart2017}
	J.~Burkart.
	\newblock Mathematica Markov chain Monte Carlo.
	\newblock 
	\url{https://github.com/joshburkart/mathematica-mcmc}, 2017.
	
	
	
	
	
	\bibitem{Cang2013-dw}
	Jianhua Cang and David~A Feldheim.
	\newblock Developmental mechanisms of topographic map formation and alignment.
	\newblock {\em Annu. Rev. Neurosci.}, 36:51--77,  2013.
	\url{https://doi.org/10.1146/annurev-neuro-062012-170341}
	
	
	
	
	\bibitem{Chandrasekaran2005-ug}
	Anand~R Chandrasekaran, Daniel~T Plas, Ernesto Gonzalez, and Michael~C Crair.
	\newblock Evidence for an instructive role of retinal activity in retinotopic
	map refinement in the superior colliculus of the mouse.
	\newblock {\em J. Neurosci.}, 25(29):6929--6938,  2005.
	\url{https://doi.org/10.1523/JNEUROSCI.1470-05.2005}
	
\bibitem{Citri2008-kv}
Ami Citri and Robert~C Malenka.
\newblock Synaptic plasticity: multiple forms, functions, and mechanisms.
\newblock {\em Neuropsychopharmacology}, 33(1):18--41,  2008.
\url{https://doi.org/10.1038/sj.npp.1301559}





		\bibitem{Coombes2005-mt}
	S~Coombes.
	\newblock Waves, bumps, and patterns in neural field theories.
	\newblock {\em Biol. Cybern.}, 93(2):91--108,  2005.
	\url{https://doi.org/10.1007/s00422-005-0574-y}
	
	
	
	\bibitem{Demas2003-xf}
	Jay Demas, Stephen~J Eglen, and Rachel O~L Wong.
	\newblock Developmental loss of synchronous spontaneous activity in the mouse
	retina is independent of visual experience.
	\newblock {\em J. Neurosci.}, 23(7):2851--2860,  2003.
	\url{https://doi.org/10.1523/JNEUROSCI.23-07-02851.2003}



\bibitem{Detorakis2012-eh}
Georgios~Is Detorakis and Nicolas~P Rougier.
\newblock A neural field model of the somatosensory cortex: formation,
maintenance and reorganization of ordered topographic maps.
\newblock {\em PLoS One}, 7(7):e40257,  2012.
\url{https://doi.org/10.1371/journal.pone.0040257}


	\bibitem{Dhande2011-jp}
Onkar~S Dhande, Ethan~W Hua, Emily Guh, Jonathan Yeh, Shivani Bhatt, Yueyi
Zhang, Edward~S Ruthazer, Marla~B Feller, and Michael~C Crair.
\newblock Development of single retinofugal axon arbors in normal and $\beta$2
{Knock-Out} mice.
\newblock {\em J. Neurosci.}, 31(9):3384--3399,  2011.
\url{https://doi.org/10.1523/JNEUROSCI.4899-10.2011}

\bibitem{Froemke2002-be}
Robert~C Froemke and Yang Dan.
\newblock Spike-timing-dependent synaptic modification induced by natural spike
trains.
\newblock {\em Nature}, 416(6879):433--438,  2002.
\url{https://doi.org/10.1038/416433a}


	\bibitem{Gelman1992-pk}
Andrew Gelman and Donald~B Rubin.
\newblock Inference from iterative simulation using multiple sequences.
\newblock {\em SSO Schweiz. Monatsschr. Zahnheilkd.}, 7(4):457--472, 
1992.
\url{https://doi.org/10.1214/ss/1177011136}

	
	
	
	\bibitem{Godfrey2009-ts}
	Keith~B Godfrey, Stephen~J Eglen, and Nicholas~V Swindale.
	\newblock A Multi-Component Model of the Developing Retinocollicular Pathway
	Incorporating Axonal and Synaptic Growth.
	\newblock {\em PLoS Comput. Biol.}, 5(12):e1000600,  2009.
	\url{https://doi.org/10.1371/journal.pcbi.1000600}
	
	\bibitem{Godfrey2009-rs}
	Keith~B Godfrey and Stephen~J Eglen.
	\newblock Theoretical models of spontaneous activity generation and propagation
	in the developing retina.
	\newblock {\em Mol. Biosyst.}, 5(12):1527--1535,  2009.
	\url{https://doi.org/10.1039/B907213F}
	
	\bibitem{Graben2014-pm}
	Peter~Beim Graben and Axel Hutt.
	\newblock Attractor and saddle node dynamics in heterogeneous neural fields.
	\newblock {\em EPJ Nonlinear Biomedical Physics}, 2(1):4,  2014.
	\url{https://doi.org/10.1140/epjnbp17}
	
	
	
	\bibitem{Hjorth2015-le}
	J~J~Johannes Hjorth, David~C Sterratt, Catherine~S Cutts, David~J Willshaw, and
	Stephen~J Eglen.
	\newblock Quantitative assessment of computational models for retinotopic map
	formation.
	\newblock {\em Dev. Neurobiol.}, 75(6):641--666,  2015.
	\url{https://doi.org/10.1002/dneu.22241}
	
	
	\bibitem{Ito2018-ef}
	Shinya Ito and David~A Feldheim.
	\newblock The mouse superior colliculus: An emerging model for studying circuit
	formation and function.
	\newblock {\em Front. Neural Circuits}, 12:10,  2018.
	\url{https://doi.org/10.3389/fncir.2018.00010}
	
	
	\bibitem{Kaas1990-kw}
	J~H Kaas, L~A Krubitzer, Y~M Chino, A~L Langston, E~H Polley, and N~Blair.
	\newblock Reorganization of retinotopic cortical maps in adult mammals after
	lesions of the retina.
	\newblock {\em Science}, 248(4952):229--231,  1990.
	\url{ https://doi.org/10.1126/science.2326637}
	
	
	
	\bibitem{Kaas1997-oz}
	J~H Kaas.
	\newblock Topographic maps are fundamental to sensory processing.
	\newblock {\em Brain Res. Bull.}, 44(2):107--112, 1997.
	\url{https://doi.org/10.1016/S0361-9230(97)00094-4}
	
	\bibitem{Kohonen1982-nd}
	Teuvo Kohonen.
	\newblock Self-organized formation of topologically correct feature maps.
	\newblock {\em Biol. Cybern.}, 43(1):59--69,  1982.
	\url{https://doi.org/10.1007/BF00337288}
	
	
	
	
	
	\bibitem{Lyngholm2019-fs}
	Daniel Lyngholm, David~C Sterratt, J~J Johannes~Hjorth, David~J Willshaw,
	Stephen~J Eglen, and Ian~D Thompson.
	\newblock Measuring and modelling the emergence of order in the mouse
	retinocollicular projection.
	\newblock {\em bioRxiv}, page 713628,  2019.
	\url{https://doi.org/10.1101/713628}
	
	
	
	
	
	
	\bibitem{Maccione2014-ha}
	Alessandro Maccione, Matthias~H Hennig, Mauro Gandolfo, Oliver Muthmann, James
	van Coppenhagen, Stephen~J Eglen, Luca Berdondini, and Evelyne Sernagor.
	\newblock Following the ontogeny of retinal waves: pan-retinal recordings of
	population dynamics in the neonatal mouse.
	\newblock {\em J. Physiol.}, 592(7):1545--1563,  2014.
	\url{https://doi.org/10.1523/JNEUROSCI.20-20-07672.2000}
	
	\bibitem{Markram1997-ln}
	H~Markram, J~L{\"u}bke, M~Frotscher, and B~Sakmann.
	\newblock Regulation of synaptic efficacy by coincidence of postsynaptic {APs}
	and {EPSPs}.
	\newblock {\em Science}, 275(5297):213--215,  1997.
	\url{https://doi.org/10.1126/science.275.5297.213}
	
	
	\bibitem{McLaughlin2003-yy}
	Todd McLaughlin, Christine~L Torborg, Marla~B Feller, and Dennis D~M O'Leary.
	\newblock Retinotopic map refinement requires spontaneous retinal waves during
	a brief critical period of development.
	\newblock {\em Neuron}, 40(6):1147--1160,  2003.
	\url{https://doi.org/10.1016/s0896-6273(03)00790-6}
	
	\bibitem{Meister1991-mu}
	M~Meister, R~O Wong, D~A Baylor, and C~J Shatz.
	\newblock Synchronous bursts of action potentials in ganglion cells of the
	developing mammalian retina.
	\newblock {\em Science}, 252(5008):939--943,  1991.
	\url{https://doi.org/10.1126/science.2035024}
	
	
	\bibitem{Merzenich1983-uz}
	M~M Merzenich, J~H Kaas, J~Wall, R~J Nelson, M~Sur, and D~Felleman.
	\newblock Topographic reorganization of somatosensory cortical areas 3b and 1
	in adult monkeys following restricted deafferentation.
	\newblock {\em Neuroscience}, 8(1):33--55,  1983.
	\url{https://doi.org/10.1016/0306-4522(83)90024-6}
	
	
	
	\bibitem{Mrsic-Flogel2005-xp}
	Thomas~D Mrsic-Flogel, Sonja~B Hofer, Claire Creutzfeldt, Isabelle
	Clo{\"e}z-Tayarani, Jean-Pierre Changeux, Tobias Bonhoeffer, and Mark
	H{\"u}bener.
	\newblock Altered map of visual space in the superior colliculus of mice
	lacking early retinal waves.
	\newblock {\em J. Neurosci.}, 25(29):6921--6928,  2005.
	\url{https://doi.org/10.1523/JNEUROSCI.1555-05.2005}
	
	
		\bibitem{Phongphanphanee2014-in}
	Penphimon Phongphanphanee, Robert~A Marino, Katsuyuki Kaneda, Yuchio Yanagawa,
	Douglas~P Munoz, and Tadashi Isa.
	\newblock Distinct local circuit properties of the superficial and intermediate
	layers of the rodent superior colliculus.
	\newblock {\em Eur. J. Neurosci.}, 40(2):2329--2343,  2014.
	\url{https://doi.org/10.1111/ejn.12579}
	
	
	
	\bibitem{Robertson1989-sm}
	D~Robertson and D~R Irvine.
	\newblock Plasticity of frequency organization in auditory cortex of guinea
	pigs with partial unilateral deafness.
	\newblock {\em J. Comp. Neurol.}, 282(3):456--471,  1989.
	\url{ https://doi.org/10.1002/cne.902820311}
	
	
		\bibitem{Robinson2005-en}
	P~A Robinson, C~J Rennie, D~L Rowe, S~C O'Connor, and E~Gordon.
	\newblock Multiscale brain modelling.
	\newblock {\em Philos. Trans. R. Soc. Lond. B Biol. Sci.},
	360(1457):1043--1050,  2005.
	\url{https://doi.org/10.1098/rstb.2005.1638}
	
	
		\bibitem{Robinson2011-ve}
	P~A Robinson.
	\newblock Neural field theory of synaptic plasticity.
	\newblock {\em J. Theor. Biol.}, 285(1):156--163,  2011.
	\url{https://doi.org/10.1016/j.jtbi.2011.06.023}
	
	
	
	\bibitem{Schwappach2015-jy}
	Cordula Schwappach, Axel Hutt, and Peter Beim~Graben.
	\newblock Metastable dynamics in heterogeneous neural fields.
	\newblock {\em Front. Syst. Neurosci.}, 9:97,  2015.
	\url{https://doi.org/10.3389/fnsys.2015.00097}
	
	
	\bibitem{Seabrook2017-kx}
	Tania~A Seabrook, Timothy~J Burbridge, Michael~C Crair, and Andrew~D Huberman.
	\newblock Architecture, function, and assembly of the mouse visual system.
	\newblock {\em Annu. Rev. Neurosci.}, 40:499--538,  2017.
	\url{https://doi.org/10.1146/annurev-neuro-071714-033842}
	
	
		\bibitem{Shi2018-er}
	Xuefeng Shi, Yanjiao Jin, and Jianhua Cang.
	\newblock Transformation of feature selectivity from membrane potential to
	spikes in the mouse superior colliculus.
	\newblock {\em Front. Cell. Neurosci.}, 12:163,  2018.
	\url{https://doi.org/10.3389/fncel.2018.00163}
	
		\bibitem{Sirosh1994-zv}
	Joseph Sirosh and Risto Miikkulainen.
	\newblock Cooperative self-organization of afferent and lateral connections in
	cortical maps.
	\newblock {\em Biol. Cybern.}, 71(1):65--78,  1994.
	\url{https://doi.org/10.1007/BF00198912}
	
	
	\bibitem{Stafford2009}
	Ben~K Stafford, Alexander Sher, Alan~M Litke, and David~A Feldheim.
	\newblock Spatial-temporal patterns of retinal waves underlying
	activity-dependent refinement of retinofugal projections.
	\newblock {\em Neuron}, 64(2):200--212,  2009.
	\url{https://doi.org/10.1016/j.neuron.2009.09.021}
	
	\bibitem{Stevens2013-ly}
	Jean-Luc~R Stevens, Judith~S Law, J{\'a}n Antol{\'\i}k, and James~A Bednar.
	\newblock Mechanisms for stable, robust, and adaptive development of
	orientation maps in the primary visual cortex.
	\newblock {\em J. Neurosci.}, 33(40):15747--15766,  2013.
	\url{https://doi.org/10.1523/JNEUROSCI.1037-13.2013}
	
	
	\bibitem{Takeuchi1979-zy}
	A~Takeuchi and S~Amari.
	\newblock Formation of topographic maps and columnar microstructures in nerve
	fields.
	\newblock {\em Biol. Cybern.}, 35(2):63--72,  1979.
	\url{https://doi.org/10.1007/BF00337432}
	
	
	
	\bibitem{Tikidji-Hamburyan2016-sn}
	Ruben~A Tikidji-Hamburyan, Tarek~A El-Ghazawi, and Jason~W Triplett.
	\newblock Novel models of visual topographic map alignment in the superior
	colliculus.
	\newblock {\em PLoS Comput. Biol.}, 12(12):e1005315,  2016.
	\url{https://doi.org/10.1371/journal.pcbi.1005315}
	
	
	\bibitem{Triplett2011-jk}
	Jason~W Triplett, Cory Pfeiffenberger, Jena Yamada, Ben~K Stafford, Neal~T
	Sweeney, Alan~M Litke, Alexander Sher, Alexei~A Koulakov, and David~A
	Feldheim.
	\newblock Competition is a driving force in topographic mapping.
	\newblock {\em Proc. Natl. Acad. Sci. U. S. A.}, 108(47):19060--19065, 
	2011.
	\url{https://doi.org/10.1073/pnas.1102834108}
	
	
	
	
	\bibitem{Tsigankov2006-uy}
	Dmitry~N Tsigankov and Alexei~A Koulakov.
	\newblock A unifying model for activity-dependent and activity-independent
	mechanisms predicts complete structure of topographic maps in ephrin-a
	deficient mice.
	\newblock {\em J. Comput. Neurosci.}, 21(1):101--114,  2006.
	\url{https://doi.org/10.1007/s10827-006-9575-7}
	
	
		\bibitem{Udin1988-by}
	S~B Udin and J~W Fawcett.
	\newblock Formation of topographic maps.
	\newblock {\em Annu. Rev. Neurosci.}, 11:289--327, 1988.
	\url{https://doi.org/10.1146/annurev.ne.11.030188.001445}
	
	
	
	\bibitem{Willshaw1976-ew}
	D~J Willshaw and C~Von~Der Malsburg.
	\newblock How patterned neural connections can be set up by self-organization.
	\newblock {\em Proc. R. Soc. Lond. B Biol. Sci.}, 194(1117):431--445, 
	1976.
	\url{https://doi.org/10.1098/rspb.1976.0087}
	
	
		\bibitem{Wright2013-td}
	James~Joseph Wright and Paul~David Bourke.
	\newblock On the dynamics of cortical development: synchrony and synaptic
	self-organization.
	\newblock {\em Front. Comput. Neurosci.}, 7:4,  2013.
	\url{https://doi.org/10.3389/fncom.2013.00004}
	
	
	
	\bibitem{Xu2015-uc}
	Hong-Ping Xu, Timothy~J Burbridge, Ming-Gang Chen, Xinxin Ge, Yueyi Zhang,
	Zhimin~Jimmy Zhou, and Michael~C Crair.
	\newblock Spatial pattern of spontaneous retinal waves instructs retinotopic
	map refinement more than activity frequency.
	\newblock {\em Dev. Neurobiol.}, 75(6):621--640,  2015.
	\url{https://doi.org/10.1002/dneu.22288}
	
	\bibitem{Zhang2000-lb}
	L~I Zhang, H~W Tao, and M~Poo.
	\newblock Visual input induces long-term potentiation of developing
	retinotectal synapses.
	\newblock {\em Nat. Neurosci.}, 3(7):708--715,  2000.
	\url{https://doi.org/10.1038/76665}
	
	
	
	
	
\end{thebibliography}


\providecommand{\bysame}{\leavevmode\hbox to3em{\hrulefill}\thinspace}
\providecommand{\MR}{\relax\ifhmode\unskip\space\fi MR }
\providecommand{\MRhref}[2]{%
  \href{http://www.ams.org/mathscinet-getitem?mr=#1}{#2}
}
\providecommand{\href}[2]{#2}

\ACKNO{This research was funded in part by the Wellcome Trust (Grant Reference: 215153/Z/18/Z) and in part by the Australian Postgraduate Award. For the purpose of open access, the author has applied a CC BY public copyright licence to any Author Accepted Manuscript version arising from this submission.}

\end{document}